 \pdfoutput=1 

\documentclass[10pt,reprint]{socc14}

\usepackage{amsmath}
\usepackage{graphicx}
\usepackage{epsfig}
\usepackage{subcaption}
\usepackage{algorithm}
\usepackage[noend]{algorithmic}
\usepackage{multirow}
\usepackage{times}
\usepackage{dsfont}
\usepackage{url}
\usepackage{color}
\usepackage{epstopdf}

\newtheorem{definition}{Definition}
\newtheorem{theorem}{Theorem}
\newtheorem{lemma}{Lemma}
\newtheorem{proof}{Proof}
\newcommand{\aries}{{\tt A\small{RIES}}}
\newcommand{\nexists}{\not\exists}

\begin{document}

\setlength{\pdfpageheight}{\paperheight}
\setlength{\pdfpagewidth}{\paperwidth}

\conferenceinfo{Submission to SoCC '15}{August, 2015, Hawaii, USA}
\copyrightyear{2015}
\copyrightdata{978-1-nnnn-nnnn-n/yy/mm}
\doi{nnnnnnn.nnnnnnn}





\title{Adaptive Logging for Distributed In-memory Databases}
\authorinfo{Chang Yao$^\ddag$, Divyakant Agrawal$^\sharp$, Gang Chen$^\S$, Beng Chin Ooi$^\ddag$, Sai Wu$^\S$}
           {$^\ddag$National University of Singapore, $^\sharp$University of California at Santa Barbara, $^\S$Zhejiang University}
           {$^\ddag$\{yaochang,ooibc\}@comp.nus.edu.sg,
           $^\sharp$agrawal@cs.ucsb.edu,
           $^\S$\{cg,wusai\}@cs.zju.edu.cn}

\maketitle
\begin{abstract}
A new type of logs, the
command log,
is being employed to replace the traditional data log (e.g., \aries\ log) in the in-memory databases.
Instead of recording how the tuples are updated,
a command log only tracks the
transactions being executed,  thereby
effectively reducing the size of the log and improving the performance.
Command logging on the other hand increases the cost of recovery,
because all the transactions in the log after the last checkpoint must be completely redone in case of a failure.
In this paper,
we first extend the command logging technique to a distributed environment,
where all the nodes can perform recovery in parallel.
We then propose an adaptive logging approach by combining data logging and
command logging.
The percentage of data logging versus command logging becomes
an optimization between the performance of transaction processing and recovery to suit
different OLTP applications.
Our experimental study compares the performance
of our proposed adaptive logging,
\aries-style data logging and command logging on top of H-Store.
The results show that adaptive logging can achieve a 10x boost for recovery and a
transaction throughput that is comparable to that of command logging.
\end{abstract}




\section{Introduction}


Harizopoulos et al.~\cite{Harizopoulos:2008:OTL:1376616.1376713} show that
in in-memory databases,
substantial amount of time is spent in logging,
latching, locking, index maintenance,
and buffer management.
The existing techniques in relational databases will lead
to suboptimal performance for in-memory databases,
 because
the assumption of I/O being the main bottleneck is no longer valid.
For instance,
in conventional databases,
the most widely used
logging approach is the write-ahead log (e.g., \aries\ log \cite{DBLP:journals/tods/MohanHLPS92}).
Write-ahead logging records the history of transactional updates to the data tuples,
and we shall refer to it as the \emph{data log} in this paper.

\begin{figure}
\centering
\includegraphics[width=0.45\textwidth,height=0.22\textwidth]{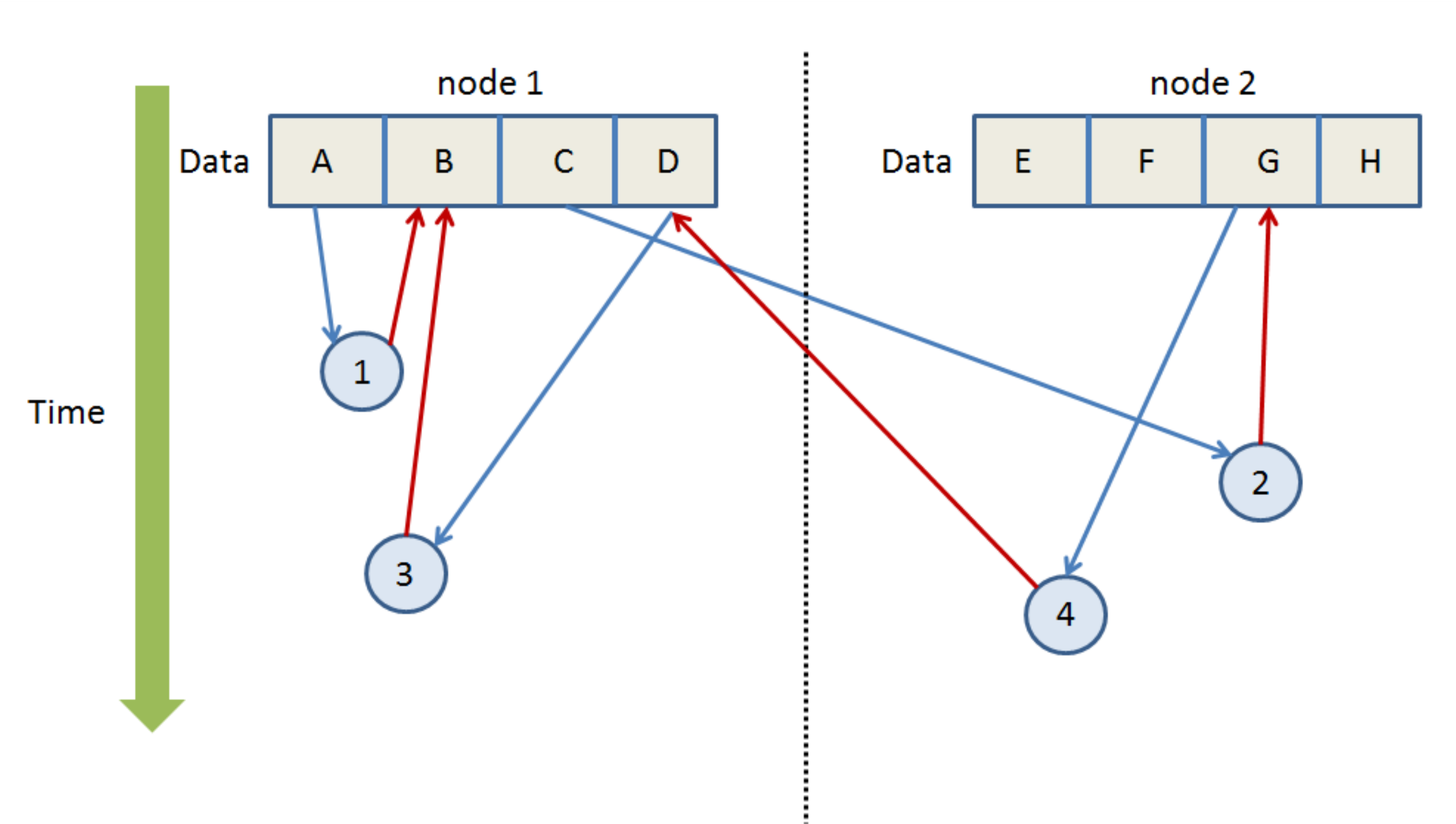}
\caption{Example of logging techniques}\label{fig:example}
\vspace {-10pt}
\end{figure}

Consider an example shown in Figure~\ref{fig:example}.
There are two nodes processing four concurrent transactions,
$T_1$ to
$T_4$.
All the transactions follow the same format:
$$f(x, y): y = 2x$$
So $T_1=f(A, B)$, indicating that $T_1$ reads the value of $A$ and then updates $B$ as $2A$.
Since different transactions may modify the same value,
there should be a locking mechanism.
Based on their timestamps,
the correct serialized order of the transactions is $T_1$, $T_2$, $T_3$ and $T_4$.
Let $v(X)$ denote the value of
parameter $X$.
The \aries\ data log of the four transactions are listed as below:

\begin{table}[h]
\vspace {-5pt}
\scriptsize
\caption{\aries\ log}\label{tb:aries}
\centering
\begin{tabular}{|l|l|l|l|l|}
\hline
\scriptsize timestamp & \scriptsize transaction ID & \scriptsize parameter & \scriptsize old value & \scriptsize new value \\\hline
100001 & $T_1$ & B & v(B) & 2v(A) \\\hline
100002 & $T_2$ & G & v(G) & 2v(C) \\\hline
100003 & $T_3$ & B & v(B) & 2v(D)\\\hline
100004 & $T_4$ & D & v(D) & 2v(G) \\\hline
\end{tabular}\\
\vspace {-5pt}
\end{table}

\aries\  log  records how the data are modified by the transactions,
and by using the log data, we can efficiently recover if there is a node failure.
However, the recovery process of in-memory databases is sightly different from that of
conventional disk-based databases.
To recover,
an in-memory database first loads the database snapshot recorded in the last
checkpoint and then replays all the committed transactions in \aries\ log.
{\color{black}
For uncommitted transactions,
no roll-backs are required, since uncommitted writes will 
not be reflected onto disk.
}

\aries\  log is a ``heavy-weight'' logging approach, as
it incurs high overheads.
In conventional database systems,
where I/Os for processing transactions
dominate the performance,
the logging cost is tolerable.
However, in an in-memory system,
since all the transactions are processed in memory,
logging cost becomes a dominant cost.

{\color{black}
To reduce the logging overhead,
a command logging approach \cite{DBLP:conf/icde/MalviyaWMS14} was
proposed to only record the transaction information with which
transaction can be fully replayed when facing a failure.
In H-Store \cite{DBLP:journals/pvldb/KallmanKNPRZJMSZHA08},
each command log records the ID of the corresponding transaction and which
stored procedure is applied to update the database along with input parameters.
As an example,
the command logging records for Figure \ref{fig:example}
are simplified as below:

\begin{table}[h]
\vspace {-5pt}
\scriptsize
\caption{Command log}\label{tb:command}
\centering
\begin{tabular}{|c|c|c|c|}
\hline
\scriptsize transaction ID &\scriptsize  timestamp & \scriptsize  procedure pointer & \scriptsize  parameters \\\hline
1 & 100001 & $p$ & A, B \\\hline
2 & 100002 & $p$ & C, G \\\hline
3 & 100003 & $p$ & D, B \\\hline
4 & 100004 & $p$ & G, D \\\hline
\end{tabular}\\
\vspace {-5pt}
\end{table}

As all four transactions follow the same routine,
we only keep a pointer $p$ to the
details of
{\color{black} storage procedure: $f(x, y): y =2x$}.
For recovery purposes,
we also need to maintain the parameters
for each transaction, so that the system can re-execute
all the transactions from the last checkpoint
when a failure happens.
Compared to \aries-style log,
a command log is much more compact and hence reduces the I/O
cost for materializing it onto disk.
It was shown that command logging can
significantly increase the throughput of transaction processing in in-memory
databases \cite{DBLP:conf/icde/MalviyaWMS14}.
However, the improvement is achieved at the expense of
its recovery performance.

When there is a node failure,
all the transactions have to be replayed in the command logging approach,
while \aries-style logging simply recovers the value of each column 
({\color{black}Note that throughout the paper,
we use "attribute" to refer to a column defined in the schema
and "attribute value" to denote the value of a tuple in a
specific column}).
For example, to fully redo $T_1$, command logging needs to read
the value of $A$ and update the value of $B$, while
if \aries\ logging is adopted,
we just set $B$'s value as $2v(A)$ as
recorded in the log file.
More importantly,
command logging does not support parallel recovery in a distributed system.
In the command logging \cite{DBLP:conf/icde/MalviyaWMS14},
command logs
of different nodes are merged at the master node during recovery, and
to guarantee the correctness of recovery,  transactions
must be reprocessed in a serialized order
 based on their timestamps.
For example, even in a network of two nodes,
the transactions have to be
replayed one by one due to their possible competition.
For the earlier example, $T_3$ and
$T_4$ cannot be concurrently processed by node 1 and 2 respectively,
because both transactions
need to lock the value of $D$. For comparison, \aries-style logging can
start the recovery in node 1 and 2 concurrently and independently.

In summary,
command logging reduces the I/O cost of processing transactions, but incurs a much
higher cost for recovery than \aries-style logging,
especially in a distributed environment.
To this end,
we propose a new logging scheme which achieves a comparable performance
as command logging for processing transactions,
while enabling a much more efficient recovery.
Our logging approach also allows the users to tune the parameters to achieve a preferable
tradeoff between transaction processing and recovery.

In this paper,
we first propose a distributed version of command logging.
In the recovery process,
before redoing the transactions,
we first generate the dependency graph by scanning the log data.
Transactions that read or write the same tuple will be linked together.
A transaction can be reprocessed only if all its dependent transactions have been processed and committed.
On the other hand,
transactions that do not have dependency relationship can be concurrently processed.
Based on this principle,
we organize transactions into different processing groups.
Transactions inside a group have dependency relationship, while transactions of
different groups can be processed concurrently.

While distributed version of command logging effectively exploits the parallelism
among the nodes to speed up recovery,
some processing groups can be rather large,
causing a few transactions to block the processing of many others.
We subsequently propose an adaptive logging approach
which adaptively makes use of the command logging and \aries-style logging.
More specifically,
we identify the bottlenecks dynamically
based on our cost model and resolve them using \aries\ logging.
We materialize the transactions identified as bottlenecks in \aries\ log.
So transactions depending on them can be recovered more efficiently.

It is indeed very challenging to classify transactions
into the ones that may cause bottleneck and those that will not,
because we have to make a real-time decision on either adopting command logging or \aries\ logging.
During transaction processing,
we do not know the impending distribution of transactions.
Even if the dependency graph of impending transactions is known before the processing starts,
we note that the optimization problem of log creation is still
an NP-hard problem.
Hence, a heuristic approach is subsequently proposed to find an approximate
solution based on our model.
The idea is to estimate the importance of each transaction
based on the access patterns of existing transactions.

Finally,
we implement our two approaches,
namely distributed command logging and adaptive logging,
on top of H-Store \cite{DBLP:journals/pvldb/KallmanKNPRZJMSZHA08} and compare
them with
\aries\ logging and command logging.
Our results show that adaptive logging can achieve a
comparable performance for transaction processing as command logging,
while it performs 10 times faster than command logging for recovery in a distributed system.

The rest of the paper is organized as follows.
We present our distributed command logging approach in Section 2
and the new adaptive logging approach in Section 3.
The experimental results are presented in
Section 4 and we review some related work in Section 5.
The paper is concluded in Section 6.


\section{Distributed Command Logging}

As described earlier,
the command logging \cite{DBLP:conf/icde/MalviyaWMS14}
only records
{\color{black} the transaction ID, storage procedure and
its input parameters}.
If some servers fail,
the database can restore the last snapshot
and redo all the transactions
in the command log to re-establish the database state.
Command logging operates
at a much coarser granularity and writes much fewer bytes per
transaction than \aries-style logging.

However,
the major concern of command logging is its recovery performance.
In VoltDB\footnote{http://voltdb.com/},
command logs of different nodes
are shuffled to the master node which merges them using the timestamp order.
Since command logging does not record how the data are manipulated,
we must redo all transactions one by one,
incurring high recovery overhead.
{\color{black}
An alternative solution is to maintain multiple replicas \cite{baker2011megastore,corbett2013spanner,rao2011using,thomson2012calvin,haughian2014benchmarking},
so that data on the failed node can be recovered
from their replicas.
However, the drawback of such approach is twofold. 
First,
keeping consistency between replicas incurs high synchronization overhead,
further slowing down the transaction processing. Second, given a limited amount of memory, it
is too expensive to maintain replicas in memory.
Therefore, in this paper, we focus on the log-based approaches,
although we also show the performance of a replication-based technique
in our experimental study.
}


Before delving into the details of our  approach,
we first define the correctness of recovery in our system.
Suppose the data are partitioned
to $N$ cluster nodes.
Let $\mathcal{T}$ be the set of transactions since the last checkpoint.
For a transaction $t_i\in\mathcal{T}$, if $t_i$ reads or writes a tuple on
node $n_x\in N$, $n_x$ becomes a participant of $t_i$. Specifically, we
use $f(t_i)$ to return all those nodes involved in $t_i$ and we
use $f^{-1}(n_x)$ to represent all the
transactions in which $n_x$ has participated.
In a distributed system,
we will assign each transaction a coordinator,
typically
the node that minimizes the data transfer for processing the transaction.
The coordinator schedules the data accesses and monitors how its transaction is
processed.
Hence, we only need to create a command log entry in the coordinator \cite{DBLP:conf/icde/MalviyaWMS14}.
We use
$\theta(t_i)$ to denote $t_i$'s coordinator. Obviously, we have $\theta(t_i) \in f(t_i)$.

Given two transactions
$t_i\in \mathcal{T}$ and $t_j\in \mathcal{T}$, we define an order
function $\prec$ as: $t_i \prec t_j$, only if $t_i$ is committed before $t_j$.
When a node $n_x$ fails, we need to redo a set of transactions $f^{-1}(n_x)$.
But these transactions may compete for the same tuple with other transactions.
Let $s(t_i)$ and $c(t_i)$ denote the submission time and commit time respectively.

{\color{black}
\begin{definition}{\bf Transaction Competition}\label{def:compete}\\
Transaction $t_i$ competes with transaction $t_j$, if
\begin{enumerate}
  \item $s(t_j)< s(t_i)<c(t_j)$.
  \item $t_i$ and $t_j$ read or write the same tuple.
\end{enumerate}
\end{definition}
}

Note that we define the competition as a unidirectional relationship.
$t_i$ competes with $t_j$,
and $t_j$ may compete with others which
may modify the same set of tuples and
commit before it.
Let $\odot(t_i)$ be the set of all the transactions that $t_i$ competes with.
We define function $g$
for transaction set $\mathcal{T}_j$ as:
$$g(\mathcal{T}_j)=\bigcup_{\forall t_i\in \mathcal{T}_j}\odot(t_i)$$
To recover the database from $n_x$'s failure,
we create an initial
recovery set $\mathcal{T}^x_0 = f^{-1}(n_x)$ and set $\mathcal{T}^x_{i+1} = \mathcal{T}^x_i \cup g(\mathcal{T}^x_i)$.
As we have a limited number of transactions,
we can find a $L$ satisfying
when $j\geq L$, we have $\mathcal{T}^x_{j+1}=\mathcal{T}^x_{j}$. This is because
there are no more transactions accessing the same set of tuples since the last checkpoint. We
call $\mathcal{T}^x_L$ the complete recovery set for $n_x$.

Finally, we define the correctness of recovery in the distributed system as:
\begin{definition}{\bf Correctness of Recovery}\\
When node $n_x$ fails, we need to redo all the transactions in its complete
recovery set by strictly following their commit order, e.g., if $t_i \prec t_j$,
then
$t_i$ must be reprocessed before $t_j$.
\end{definition}

To recover from a node's failure,
we need to retrieve its complete recovery set.
For this purpose, we build a dependency graph.

\subsection{Dependency Graph}

Dependency graph is defined as an acyclic direct graph $G=(V, E)$,
where each vertex
$v_i$ in $V$ represents a transaction $t_i$,
containing the information about its
timestamp ($c(t_i)$ and $s(t_i)$) and coordinator $\theta(t_i)$.
$v_i$ has an edge $e_{ij}$ to $v_j$, iff
\begin{enumerate}
  \item $t_i \ in \odot(t_j)$
  \item $\forall t_m\in \odot(t_j), c(t_m)<c(t_i)$
\end{enumerate}

\begin{figure}
\centering
\includegraphics[scale=0.4]{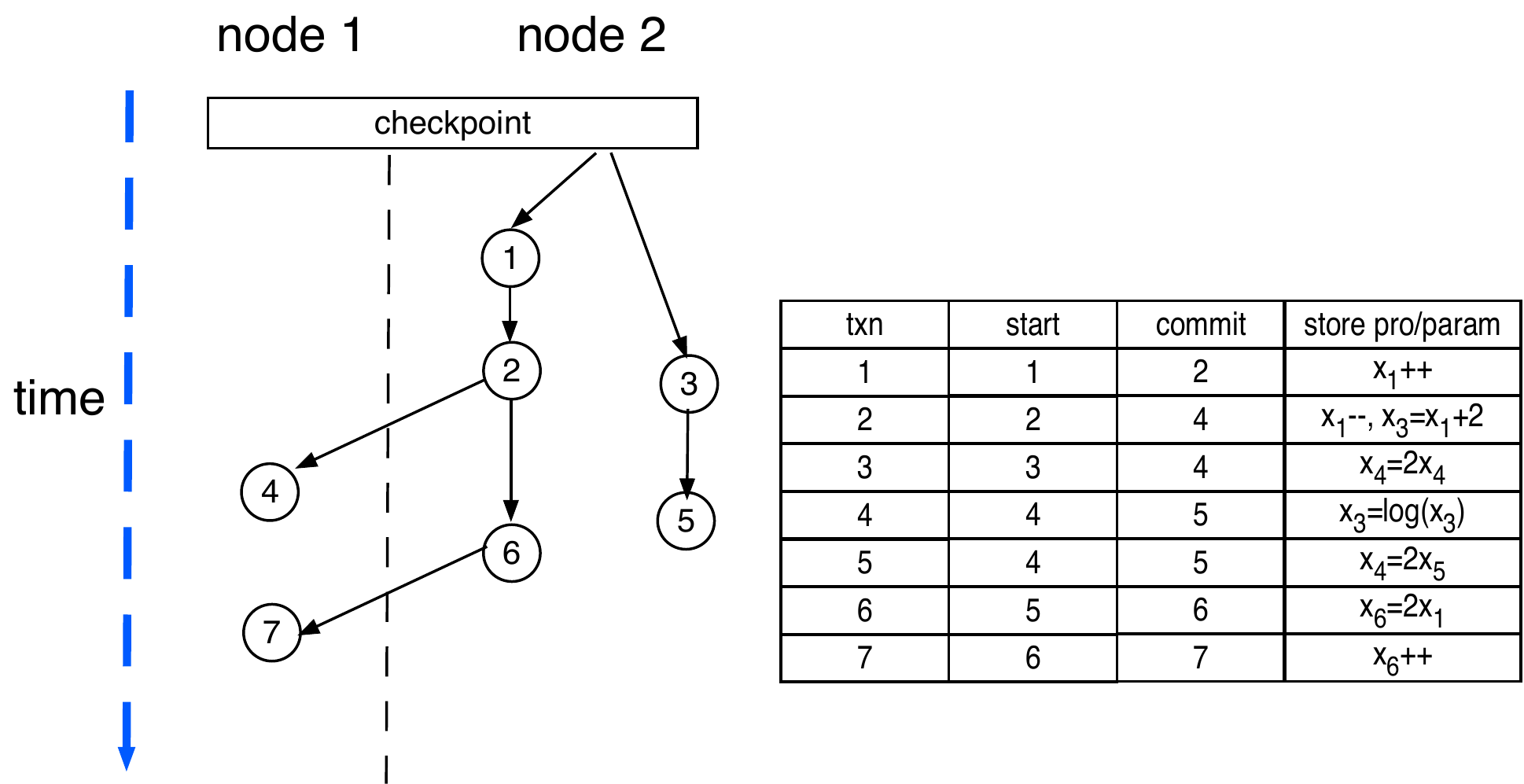}
\caption{\color{black}A running example}\label{fig:example2}
\vspace {-10pt}
\end{figure}

For a specific order of transaction processing,
there is one unique dependency graph
as shown in the following theorem.
\begin{theorem}
Given a transaction set $\mathcal{T}=\{t_0, ..., t_k\}$, where $c(t_i)<c(t_{i+1})$,
we
can generate a unique dependency graph for $\mathcal{T}$.
\end{theorem}

\begin{proof}
Since the vertices represent transactions,
we always have the same set of
vertices for the same set of transactions. We only need to prove that the edges
are also unique. Based on the definition,
 edge $e_{ij}$ exists, only if $t_j$ accesses the same set
of tuples that $t_i$ updates and no other transactions that commit
after $t_i$ have that property.
As for each transaction $t_j$, there is only one such
transaction $t_i$. Therefore, edge $e_{ij}$
is a unique edge between $t_i$ and $t_j$.
\end{proof}

{\color{black}
We use Figure \ref{fig:example2} as a running example
to illustrate the idea.
In Figure \ref{fig:example2},
there are totally seven
transactions since the last checkpoint,
aligned based on their coordinators: node 1 and node 2.
We show transaction IDs, timestamps, storage procedures and parameters in
the table.
Based on the definition, transaction $t_2$ competes with transaction $t_1$,
as both of them update $x_1$.
Transaction $t_4$ competes with transaction $t_2$ on $x_3$.
The complete recovery set for $t_4$ is $\{t_1, t_2, t_4\}$.
Note that
although $t_4$ does not access the attribute that $t_1$ updates,
$t_1$ is still in
$t_4$'s recovery set because of the recursive dependency between $t_1$, $t_2$ and $t_4$.
After constructing the dependency graph and generating the recovery set,
 we can adaptively
recover the failed node.
For example, to recover node 1, we do not have to reprocess
transactions $t_3$ and $t_5$.
}

\subsection{Processing Group}

In order to generate the complete recovery set efficiently,
we organize transactions as
processing groups.
Algorithm \ref{algo:group} and \ref{algo:addgroup} illustrate
how we generate the groups from a dependency graph.
In  Algorithm \ref{algo:group},
 we start from the root vertex that
represents the checkpoint to iterate all the vertices in the graph.
The neighbors
of the root vertex are
transactions that do not compete with the others.
We create one processing group for each of them (line 3-7).
$AddGroup$ is
a recursive function that explores all reachable vertices and adds them
into the group.
One transaction can exist in multiple groups if more than
one transaction competes with it.

\begin{algorithm}
\footnotesize \caption{\label{algo:group}{\tt CreateGroup(DependencyGraph $G$)}}
\begin{algorithmic}[1]
\STATE Set $S = \emptyset$
\STATE Vertex $v$ = $G$.getRoot()
\WHILE{$v$.hasMoreEdge()}
\STATE Vertex $v_0$ = $v$.getEdge().endVertex()
\STATE Group $g$ = new Group()
\STATE AddGroup($g$, $v_0$)
\STATE $S$.add($g$)
\ENDWHILE
\STATE return $S$
\end{algorithmic}
\end{algorithm}
\vspace {-10pt}

\begin{algorithm}
\footnotesize \caption{\label{algo:addgroup}{\tt AddGroup(Group $g$, Vertex $v$)}}
\begin{algorithmic}[1]
\STATE $g$.add($v$)
\WHILE{$v$.hasMoreEdge()}
\STATE Vertex $v_0$ = $v$.getEdge().endVertex()
\STATE AddGroup($g$, $v_0$)
\ENDWHILE
\end{algorithmic}
\end{algorithm}


\subsection{Algorithm for Distributed Command Logging}

When we detect that node $n_x$ fails,
we stop the transaction processing and start
the recovery process.
One new node starts up
 to reprocess all the transactions in $n_x$'s
complete recovery set.
Because some transactions are distributed transactions involving
other nodes,
the recovery algorithm runs as a distributed process.

Algorithm \ref{algo:cmdlog} shows the basic idea of recovery.
First, we retrieve all the
transactions that do not compete with the others
since the last checkpoint (line 3).
These transactions can be processed in parallel.
Therefore, we find their coordinators
and forward them correspondingly for processing.
At each coordinator, we invoke
Algorithm \ref{algo:prev} to process a specific transaction $t$.
We first wait until
all the transactions in $\odot(t)$ are processed.
Then, if $t$ has not been processed
yet,
we will process it and  retrieve all its neighbor transactions following the links
in the dependency graph.
If those transactions are also in the recovery set,
 we recursively
invoke function {\em ParallelRecovery} to process them.

\begin{algorithm}
\footnotesize \caption{\label{algo:cmdlog}{\tt Recover(Node $n_x$, DependencyGraph $G$)}}
\begin{algorithmic}[1]
\STATE Set $S_T$ = getAllTransactions($n_x$)
\STATE CompleteRecoverySet $S$=getRecoverySet($G$,$S_T$)
\STATE Set $S_R$ = getRootTransactions($S$)
\FOR{Transaction $t \in S_R$}
\STATE Node $n$ = $t$.getCoordinator()
\STATE ParallelRecovery($n$, $S_T$, $t$)
\ENDFOR
\end{algorithmic}
\end{algorithm}

\begin{algorithm}
\footnotesize \caption{\label{algo:prev}{\tt ParallelRecovery(Node $n$, Set $S_T$, Transaction $t$)}}
\begin{algorithmic}[1]
\WHILE{wait($\odot(t)$)}
\STATE sleep(timethreshold)
\ENDWHILE
\IF{$t$ has not been processed}
\STATE process($t$)
\STATE Set $S_t$ = $g$.getDescendant($t$)
\FOR{$\forall t_i \in S_t\cap S_T$}
\STATE Node $n_i$ = $t$.getCoordinator()
\STATE ParallelRecovery($n_i$, $S_T$, $t_i$)
\ENDFOR
\ENDIF
\end{algorithmic}
\end{algorithm}

\begin{theorem}
Algorithm \ref{algo:cmdlog} guarantees the correctness of the recovery.
\end{theorem}

\begin{proof}
In Algorithm \ref{algo:cmdlog},
if two transactions $t_i$ and $t_j$
are in the same processing group and $c(t_i)<c(t_j)$, $t_i$ must
be processed before $t_j$, as we follow the links of dependency
graph.
The complete recovery set of $t_j$ is the subset of the union of all the
processing
groups that $t_j$ joins.
Therefore, we will redo all the transactions in the recovery
set for a specific transaction as in Algorithm \ref{algo:cmdlog}.
\end{proof}

As an example, suppose node 1 fails in Figure \ref{fig:example2}. The
recovery set is $\{t_1, t_2, t_4, t_6, t_7\}$.
We will first redo $t_1$ in node 2
which is the only transaction that can run
without waiting for the other transactions.
Note that although node 2 does not fail,
we still need to reprocess $t_1$, because it modifies the tuples that are accessed
by those failed transactions.
After $t_1$ and $t_2$ commit,
we will ask the new node
which replaces node 1 to reprocess $t_4$.
 Simultaneously, node 2 will process $t_6$
in order to recover $t_7$.

\subsection{Footprint of Transactions}

To reduce the overhead of transaction processing,
a dependency graph is built offline.
Before a recovery process starts,
we scan the log to build the dependency graph.
For this purpose,
we introduce a light weight
footprint for transactions.
Footprint is a specific type of write ahead log.
Once a
transaction is committed, we record the transaction ID and the involved tuple ID as
its footprint.
Figure \ref{fig:vs} illustrates the structures of footprint and \aries\ log.
\aries\ log
maintains detailed information about a transaction,
including  partition ID,
table name, modified column, original value and updated value,
based on which we
can successfully redo a transaction.
On the contrary, footprint only records IDs of those
tuples that are read or updated by a transaction.
It incurs much less storage overhead
than \aries\ log (on average,
each record in \aries\ log and footprint requires 3KB and 450B respectively)
and hence,
does not significantly affect the performance of transaction processing.
The objective of recording footprints
is not to recover lost transactions,
but to build the dependency graph.

\begin{figure}
\centering
\includegraphics[scale=0.65]{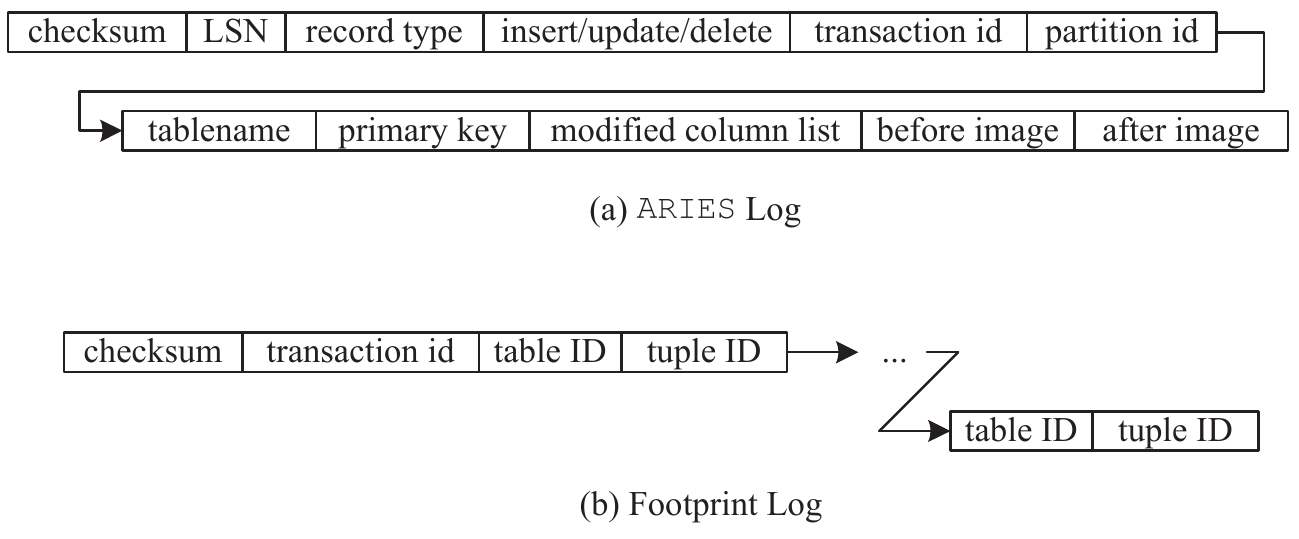}
\caption{\aries\ Log VS Footprint Log}\label{fig:vs}
\vspace {-15pt}
\end{figure}

\section{Adaptive Logging}

The bottleneck of our
distributed command logging is caused by dependencies among transactions.
To ensure causal consistency \cite{Bailis:2013:BCC:2463676.2465279},
transaction $t$ is blocked until
all the transactions in $\odot(t)$ have been processed.
If we fully or partially resolve dependencies among transactions,
the overhead of recovery can be effectively reduced.

\subsection{Basic Idea}
\label{sec:basic_idea}
As noted in the introduction,
\aries\ log allows each node to recover independently.
If node $n_x$ fails,
we just load its log data since the last checkpoint
to redo all updates.
We do not need to consider
the dependencies among transactions,
because the log completely records how a transaction
modifies the data.
Hence, the intuition of our adaptive logging approach is to combine
command logging
and \aries\ logging.
For transactions highly dependent on the others,
we create \aries\ log for these transactions to speed up their reprocessing.
For other transactions,
we apply command logging to reduce logging overhead.

{\color{black}
For example, in Figure \ref{fig:example2},
if we create \aries\ log for $t_7$,
we do not need to reprocess $t_6$ to recover node 1.
Moreover,
if \aries\ log has been created for $t_2$,
we just need to redo $t_2$ and
then $t_4$, and
the recovery process does not need to start from $t_1$.
In this case, to recover
node 1, 
only three transactions need to be re-executed, namely $\{t_2, t_4, t_7\}$.
To determine whether a transaction depends on the results of other
transactions,
we need a new relationship other than the transaction
competition that
describes the causal consistency.
}

\begin{definition}{\bf Time-Dependent Transaction}\\
Transaction $t_j$ is $t_i$'s time-dependent transaction, if
1) $c(t_i)>c(t_j)$; 2) $t_j$ updates tuple an attribute $a_x$
of tuple $r$ which is accessed by $t_i$; and 3) there is no
other transaction with commit time between $c(t_i)$ and $c(t_j)$
which also updates $r.a_x$.
\end{definition}

Let $\otimes(t_i)$ denote all $t_i$'s time-dependent transactions.
For transactions in $\otimes(t_i)$,
we can recursively find their own time-dependent transactions,
denoted as $\otimes^2(t_i)=\otimes(\otimes(t_i))$.
This process continues until we find the minimal $x$ satisfying
$\otimes^x(t_i)=\otimes^{x+1}(t_i)$.
$\otimes^x(t_i)$ represents all transactions
that must run before $t_i$ to guarantee the causal consistency.
For a special case,
if transaction $t_i$ does not compete with the others,
it does not have time-dependent transactions (namely, $\otimes(t_i)=\emptyset$) either.
$\otimes^x(t_i)$ is a subset of the complete recovery set of $t_i$.
Instead of redoing all the transactions in the complete recovery set,
we only need to process those in $\otimes^x(t_i)$ to guarantee that
$t_i$ can be recovered correctly.

If we adaptively select some transactions in $\otimes^x(t_i)$ to
create \aries\ logs,
we can effectively reduce
the recovery overhead of $t_i$.
That is,
if we have created
\aries\ log for transaction $t_j$,
$\otimes(t_j)=\emptyset$ and $\otimes^x(t_j)=\emptyset$,
because $t_j$ now can recover by simply loading its \aries\ log (in other words,
it does not depend on the results of the other transactions).

More specifically,
let $A=\{a_0, a_1,...,a_m\}$
denote the attributes that $t_i$ needs to access.
These attributes may
come from different tuples.
We use $\otimes(t_i.a_x)$ to represent
the time-dependent transactions that have updated $a_x$.
Therefore, $\otimes(t_i)=\otimes(t_i.a_0)\cup...\cup \otimes(t_i.a_m)$.
To formalize how \aries\ log can reduce the recovery overhead, we introduce
the following lemmas.

\begin{lemma}
If we have created an \aries\
log for $t_j\in \otimes(t_i)$,
transactions $t_l$ in $\otimes^{x-1}(t_j)$ can be discarded
from $\otimes^x(t_i)$, if
$$\nexists t_m\in \otimes(t_i), t_m=t_l \vee t_l \in \otimes^{x-1}(t_m)$$
\end{lemma}

\begin{proof}
The lemma indicates that all the time-dependent transactions of
$t_j$ can be discarded,
if they are not time-dependent transactions
of the other transactions in $\otimes(t_i)$, which is obviously true.
\end{proof}

The above lemma can be further extended for a random transaction in $\otimes^x(t_i)$.
\begin{lemma}
Suppose we have created an \aries\ log
for transaction $t_j\in \otimes^x(t_i)$ which
updates attribute set $\bar{A}$.
Transaction
$t_l\in \otimes^x(t_j)$ can be discarded, if
$$\nexists a_x \in (A-\bar{A}), t_l\in\otimes^x(t_i.a_x) $$
\end{lemma}

\begin{proof}
Because $t_j$ updates $\bar{A}$,
all the transactions in $\otimes^x(t_j)$
that only update attribute values in $\bar{A}$ can be discarded without
violating the correctness of casual consistency.
\end{proof}

The lemma shows that all $t_j$'s time-dependent transactions are not necessary in
the recovery process, if they are not time-dependent transactions of any attribute
in $(A-\bar{A})$.
To recover the values of attribute set $\bar{A}$ for $t_i$,
we can start from $t_j$'s \aries\ log to redo
$t_j$ and then all transactions which also update $\bar{A}$ and have
timestamps in the range of $(c(t_j), c(t_i))$.
To simplify the presentation,
we use $\phi(t_j, t_i, t_j.\bar{A})$ to denote these transactions.

Finally, we summarize our observations as the following theorem,
based on which
we design our adaptive logging and recovery algorithm.
\begin{theorem}
Suppose we have created \aries\ logs for transaction set $\mathcal{T}_a$. To recover
$t_i$, we need to redo all the transactions in
$$\bigcup_{\forall a_x\in (A-\bigcup_{\forall t_j\in \mathcal{T}_a}t_j.\bar{A})} \otimes^x(t_i.a_x)  \cup \bigcup_{\forall t_j\in \mathcal{T}_a} \phi(t_j, t_i, t_j.\bar{A})$$
\end{theorem}

\begin{proof}
The first term represents all the
transactions that are required to recover attribute values
in $(A-\bigcup_{\forall t_j\in \mathcal{T}_a}t_j.\bar{A})$.
The second
term denotes all those transactions that we need to do by recovering from \aries\ logs
and following the timestamp order.
\end{proof}

\subsection{Logging Strategy}

By combining \aries\ logging and command logging into
a hybrid logging approach,
we
can effectively reduce the recovery cost.
Given a limited I/O budget $B_{i/o}$,
our adaptive approach selects the transactions
for \aries\ logging to maximize the recovery performance.
This decision has to be made during transaction processing, where
we determine which type of logs to create for each transaction before it commits.
However, since we do not know the future distribution of transactions,
it is impossible to generate an optimal selection.
In fact, even we know all the future transactions,
the optimization problem is still NP-Hard.

Let $w^{aries}(t_j)$ and $w^{cmd}(t_j)$ denote the I/O costs of \aries\ logging and
command logging for transaction $t_j$ respectively.
We use $r^{aries}(t_j)$ and $r^{cmd}(t_j)$ to represent the recovery cost of
$t_j$ regarding to the \aries\ logging and
command logging respectively.
If we create an \aries\
log for transaction $t_j$
that is a time-dependent transaction of $t_i$,
the recovery cost is reduced by:
\begin{equation}\label{eq:benefit}
\small
  \Delta(t_j, t_i) = \sum_{\forall t_x\in \otimes^{x}(t_i)} r^{cmd}(t_x) -
  \sum_{\forall t_x\in \phi(t_j, t_i, t_j.\bar{A})} r^{cmd}(t_x) - r^{aries}(t_j)
\end{equation}

If we decide to create \aries\ log for more
than one transaction in $\otimes^{x}(t_i)$,
$\Delta(t_j, t_i)$ should be updated accordingly.
Let $\mathcal{T}_a\subset \otimes^{x}(t_i)$ be the
transactions with \aries\ logs.
We define an attribute set:
$$p(\mathcal{T}_a, t_j)=\bigcup_{\forall t_x\in \mathcal{T}\wedge c(t_x)>c(t_j)} t_x.\bar{A}$$
$p(\mathcal{T}_a, t_j)$ represent the attributes that are updated after $t_j$ by the transactions
with \aries\ logs.
Therefore, $\Delta(t_j, t_i)$ is adjusted as
\begin{eqnarray}\label{eq:benefit2}\nonumber
  \Delta(t_j, t_i, \mathcal{T}_a) &=& \sum_{\forall t_x\in \otimes^{x}(t_i)-\mathcal{T}_a} r^{cmd}(t_x) - r^{aries}(t_j)-\\
   && \sum_{\forall t_x\in \phi(t_j, t_i, t_j.\bar{A} - p(\mathcal{T}_a, t_j))} r^{cmd}(t_x)
\end{eqnarray}

\begin{definition}{\bf Optimization Problem}\\
For transaction $t_i$, finding a transaction set $\mathcal{T}_a$ to create \aries\ logs
so that $\sum_{\forall t_j\in \mathcal{T}_a} \Delta(t_j, t_i, \mathcal{T}_a)$ is maximized with the condition
$\sum_{\forall t_j\in \mathcal{T}_a} w^{aries}(t_j)\leq B_{i/o}$.
\end{definition}
Note that this is a simplified version of optimization problem,
as we only consider
a single transaction for recovery.
In real systems, if node $n_x$ fails, all
the transactions in $f(n_x)$ should be recovered.

The single transaction case of optimization
is analogous to the 0-1 knapsack problem,
while the more general case
is similar to the multi-objective knapsack problem.
It becomes
even harder when function $\Delta$ is also determined by the
correlations of transactions.

\subsubsection{Offline Algorithm}

We first introduce our offline algorithm designed for the
ideal case,
where the impending distribution of transactions is known.
The offline algorithm is only used to demonstrate the
basic idea of adaptive logging,
while our system employs its
online variant.
We use $\mathcal{T}$ to
represent all the transactions
from the last checkpoint to the point of failure.

For each transaction $t_i \in \mathcal{T}$,
we compute its benefit
as:
\begin{equation}\nonumber
  b(t_i)=\sum_{\forall t_j\in \mathcal{T} \wedge c(t_i)<c(t_j)} \Delta(t_i, t_j, \mathcal{T}_a) \times \frac{1}{w^{aries}(t_i)}
\end{equation}
Initially, $\mathcal{T}_a=\emptyset$.

We sort the transactions based on their benefit values. The one with the
maximal benefit is selected and added to $\mathcal{T}_a$.
All the transactions update their benefits accordingly based on
Equation \ref{eq:benefit2}.
This process continues until
$$\sum_{\forall t_j\in \mathcal{T}_a} w^{aries}(t_j)\leq B_{i/o}$$.
Algorithm \ref{algo:offline} outlines
the basic idea of the offline algorithm.

\begin{algorithm}
\footnotesize \caption{\label{algo:offline}{\tt Offline(TransactionSet $\mathcal{T}$)}}
\begin{algorithmic}[1]
\STATE Set $\mathcal{T}_a=\emptyset$, Map benefits;
\FOR{$\forall t_i\in \mathcal{T}$}
\STATE benefits[$t_i$] = computeBenefit($t_i$)
\ENDFOR
\WHILE{getTotalCost($\mathcal{T}_a$)$<B_{i/o}$}
\STATE sort(benefits)
\STATE $\mathcal{T}_a$.add(benefits.keys().first())
\ENDWHILE
\STATE return $\mathcal{T}_a$
\end{algorithmic}
\end{algorithm}

Since we need to re-sort all the transactions after each update to $\mathcal{T}_a$,
the complexity of the algorithm is $O(N^2)$, where $N$ is the number of transactions.
In fact, full sorting is not necessary for most cases,
because $\Delta(t_i, t_j, \mathcal{T}_a)$
should be recalculated, only if both $t_i$ and $t_j$ update a value of the same attribute.

\subsubsection{Online Algorithm}

Our online algorithm is similar to the offline version, except that we must
choose either \aries\ logging or command logging in real-time.
Since we have no knowledge about the distribution of
future transactions, we use a histogram
to approximate the distribution.
In particular, for all the attributes $A=(a_0,...,a_k)$
involved in transactions,
we record the number of transactions that read or write a specific attribute value,
and use the histogram to estimate the probability
of accessing an attribute $a_i$, denoted as $P(a_i)$. 
Note that attributes in $A$
may come from the same tuple or different tuples.
For tuple $v_0$ and $v_1$,
if both $v_0.a_i$ and $v_1.a_i$ appear in $A$, we will represent them as two
different attributes.

{\color{black}
As defined in section~\ref{sec:basic_idea},
$\phi(t_j, t_i, t_j.\bar{A})$ denotes the
transactions that commit between $t_j$ and $t_i$ and also update some attributes in $t_j.\bar{A}$.
}
As a matter of fact, we can rewrite as:
$$\phi(t_j, t_i, t_j.\bar{A})=\bigcup_{\forall a_i\in t_j.\bar{A}}\phi(t_j, t_i, a_i)$$
Similarly, let $S=t_j.\bar{A} - p(\mathcal{T}_a, t_j)$. 
The third term of Equation \ref{eq:benefit2} can be
computed as:
\begin{eqnarray}\label{eq:benefit3}\nonumber
  \sum_{\forall t_x\in \phi(t_j, t_i, S)} r^{cmd}(t_x) &=& \sum_{\forall a_x\in S } (\sum_{\forall t_x\in \phi(t_j, t_i, a_x)} r^{cmd}(t_x))
\end{eqnarray}

We use
a constant $R^{cmd}$ to denote the average recovery cost of command logging.
The above Equation can then be simplified as:
\begin{eqnarray}\label{eq:benefit4}
  \sum_{\forall t_x\in \phi(t_j, t_i, S)} r^{cmd}(t_x) &=& \sum_{\forall a_x\in S } (P(a_x) R^{cmd})
\end{eqnarray}
The first term of Equation \ref{eq:benefit2} estimating the cost of recovering
$t_j$'s time-dependent transactions using command logging can be efficiently computed in
real-time, if we maintain the dependency graph.
Therefore, by
combining Equation \ref{eq:benefit2} and \ref{eq:benefit4}, we can estimate
the benefit  $b(t_i)$
of a specific transaction
during online processing.
Suppose we have already created \aries\ logs
for transactions in $\mathcal{T}_a$,
the benefit should be updated based on Equation \ref{eq:benefit2}.

The last problem is how to define a threshold $\gamma$.
When the benefit
of a transaction is greater than $\gamma$, we create \aries\
log for it.
Let us consider the ideal case.
Suppose the node fails while processing
$t_i$ for which we have just created its \aries\ log.
This log achieves the
maximal benefit which can be estimated as:
$$b^{opt}_i=(\mathds{N} R^{cmd}\sum_{\forall a_x\in A}P(a_x) - R^{aries}) \times \frac{1}{W^{aries}}$$
where $\mathds{N}$ denotes the number of transactions before $t_i$, $R^{aries}$
and $W^{aries}$ are the average recovery cost and I/O cost of \aries\ log respectively.

Suppose the failure happens arbitrarily following a Poisson distribution with
parameter $\lambda$.
That is,
the expected average failure time is $\lambda$.
{\color{black}
Let $\rho(s)$ be the function that returns the number of committed transactions in $s$.
Before failure,
there are approximate $\rho(\lambda)$ transactions.
}
So the
possibly maximal benefit is:
$$b^{opt}=(\rho(\lambda) R^{cmd}\sum_{\forall a_x\in A}P(a_x) - R^{aries}) \times \frac{1}{W^{aries}}$$
We define our threshold as $\gamma = \alpha b^{opt}$,
 where $\alpha$ is a tunable
parameter.

Given a limited I/O budget,
we can create approximately $\frac{B_{i/o}}{W^{aries}}$
\aries\  log records.
As failures may happen randomly at anytime,
the log should
be evenly distributed over the timeline.
More specifically,
the cumulative distribution function (CDF) of the Poisson distribution is
$$P(fail\_time<k)=e^{-\lambda}\sum_{i=0}^{\lfloor k \rfloor}\frac{\lambda^i}{i!}$$
Hence, at the $k$th second,
we can maximally create
$$quota(k)=P(fail\_time<k) \frac{B_{i/o}}{W^{aries}}$$
log records.
When time elapses,
we should check whether we still have the
quota for \aries\ log.
If not, we will not create any new \aries\ log for the time being.

Finally, we summarize the idea of
online adaptive
logging scheme in Algorithm \ref{algo:online}.

\begin{algorithm}
\footnotesize \caption{\label{algo:online}{\tt Online(Transaction $t_i$, int $usedQuota$)}}
\begin{algorithmic}[1]
\STATE int $q$ = getQuota($s(t_i)$)- $usedQuota$
\IF{$q>0$}
\STATE Benefit $b$=computeBenefit($t$)
    \IF{$b>\tau$}
    \STATE $usedQuota$++
    \STATE createAriesLog($t_i$)
    \ELSE
    \STATE createCommandLog($t_i$)
    \ENDIF
\ENDIF
\end{algorithmic}
\end{algorithm}

\subsection{In-Memory Index}

To help compute the benefit of each transaction,
we create an
in-memory inverted index in our master node.
Figure \ref{fig:index} shows the structure of the index.
The index data are organized by table ID and tuple ID.
For each specific tuple, we record the transactions that
read or write its attributes.
 As an example, in Figure \ref{fig:index},
transaction $t_2$ reads the {\em number} of tuple 10001 and updates the
{\em price} of tuple 10023.

\begin{figure}
\centering
\includegraphics[scale=0.8]{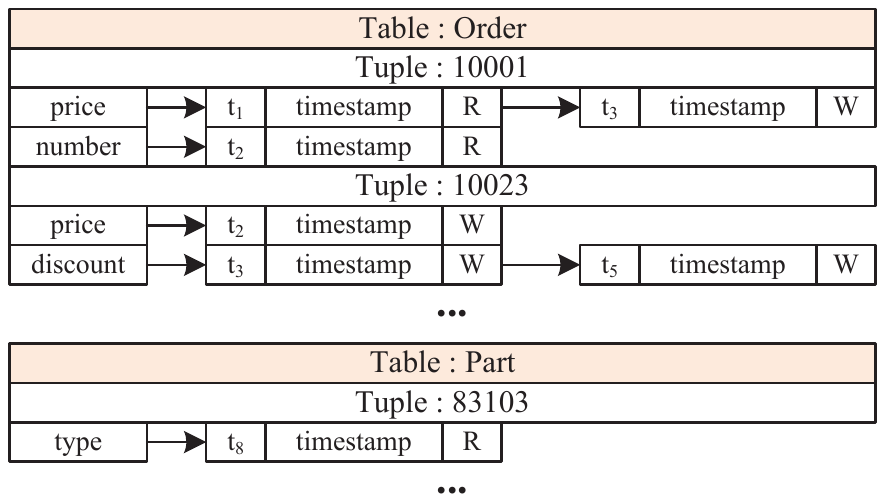}
\vspace {-10pt}
\caption{In-memory index}\label{fig:index}
\vspace {-10pt}
\end{figure}

Using the index, we can efficiently
retrieve the time-dependent transactions. For transaction $t_5$,
let $A_5$ be the attributes that it accesses. We search the
index to retrieve all transactions that update any attribute in $A_5$
before $t_5$.
In Figure \ref{fig:index},
because {\em discount} value of tuple
10023 is updated by $t_5$, we check its list and find that
$t_3$ updates the same value before $t_5$.
Therefore, $t_3$ is a
time-dependent transaction of $t_5$.
In fact, the index can also be employed to recover the
dependency graph of transactions.
We omit the details as it
is quite straightforward.

%

\section{Experimental Evaluation}
In this section,
we conduct the runtime cost analysis of our proposed adaptive logging
and compare its query processing and recovery performance 
against other approaches'.
Since both traditional \aries\ logging and command logging are already supported by
H-Store,
for consistency,
we implement distributed command logging and adaptive logging approaches on
top of the H-store as well.
In summary, we have the following four approaches:
\begin{itemize}
\item {\bf ARIES} --  \aries\ logging.
\item {\bf Command} -- command logging  proposed in \cite{DBLP:conf/icde/MalviyaWMS14}.
\item {\bf Dis-Command} -- distributed command logging approach.
\item {\bf Adapt-x\%} --  adaptive logging approach, where we create \aries\ log for x\% of distributed
transactions that involve multiple compute nodes.
{\color{black} When x=100, adaptive logging adopts a simple
strategy: \aries\ logging for all distributed transactions and command logging for all single-node transactions.}
\end{itemize}

All the experimental evaluations
are conducted on our in-house cluster of 17 nodes.
The head node is a powerful server equipped with an Intel(R) Xeon(R) 2.2 GHz 48-core CPU and  64GB RAM,
and the compute nodes are
 blade servers with an Intel(R) Xeon(R) 1.8 GHz 8-core CPU and 16GB RAM.
H-Store is deployed on the 16 compute nodes
by partitioning the databases evenly.
Each node
runs a transaction site.
By default,
only 8 sites in H-Store are used,
except in the scalability test.
We use the TPC-C benchmark\footnote{http://www.tpc.org/tpcc/},
with 100 clients being run concurrently
in the head node to submit their transaction requests one by one.
{\color{black}
As H-Store does not support replications,
we measure the effect of replication using
its commercial version VoltDB with
Voter benchmark\footnote{http://hstore.cs.brown.edu/documentation/deployment/benchmarks/voter/}\cite{DBLP:conf/icde/MalviyaWMS14}.
}
\subsection{Runtime Cost Analysis}
We first compare the overheads of different logging strategies during the runtime.
In this experiment,
we use the number of New-Order transactions processed per second
as the metric to evaluate
the effect of different
logging methods on the throughput of the system.
To illustrate the behaviors of different logging techniques,
we adopt two workloads:
one workload contains only local transactions
while the other one contains both local and distributed transactions.

\begin{figure*}[ht]
\centering
\begin{minipage}[b]{0.23\linewidth}
\includegraphics[width=1\textwidth,height=0.9\textwidth]{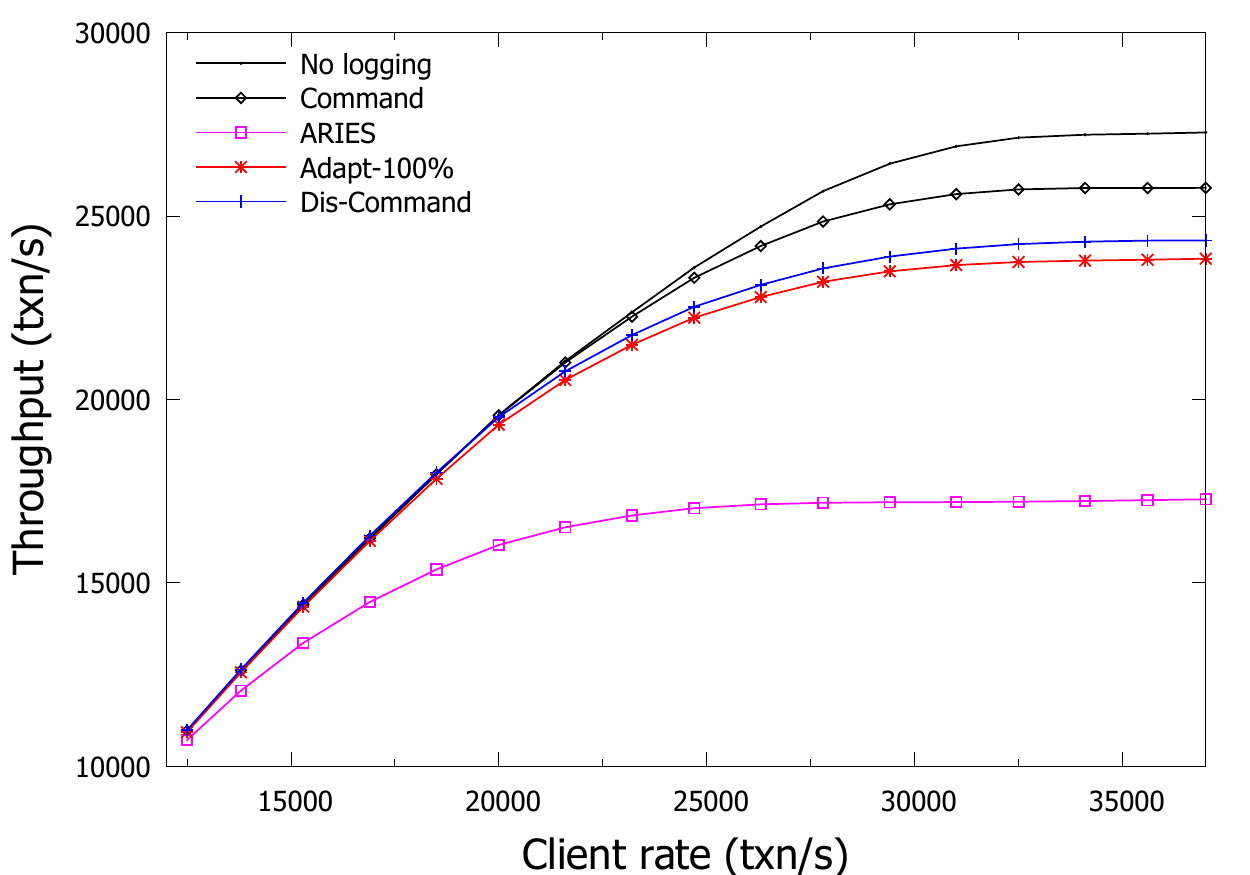}
\caption{\small Throughput without distributed transactions\\}
\label{fig:throughout_1}
\end{minipage}
\hspace{-3mm}
\hfill
\begin{minipage}[b]{0.23\linewidth}
\includegraphics[width=1\textwidth,height=0.9\textwidth]{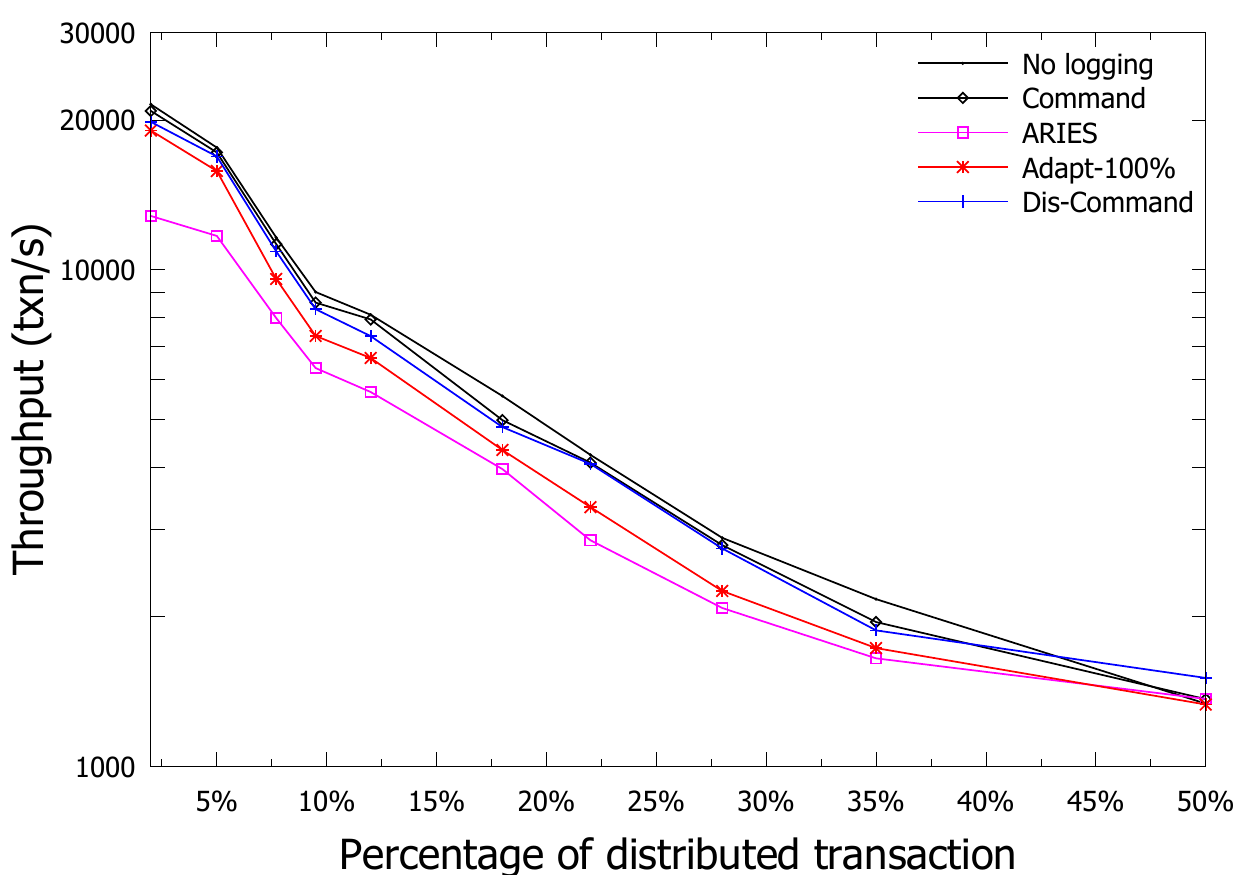}
\caption{\small Throughput with distributed transactions (with log scale on y-axis) }
\label{fig:dis_thr_1}
\end{minipage}
\hspace{-3mm}
\hfill
\begin{minipage}[b]{0.23\linewidth}
\includegraphics[width=1\textwidth,height=0.9\textwidth]{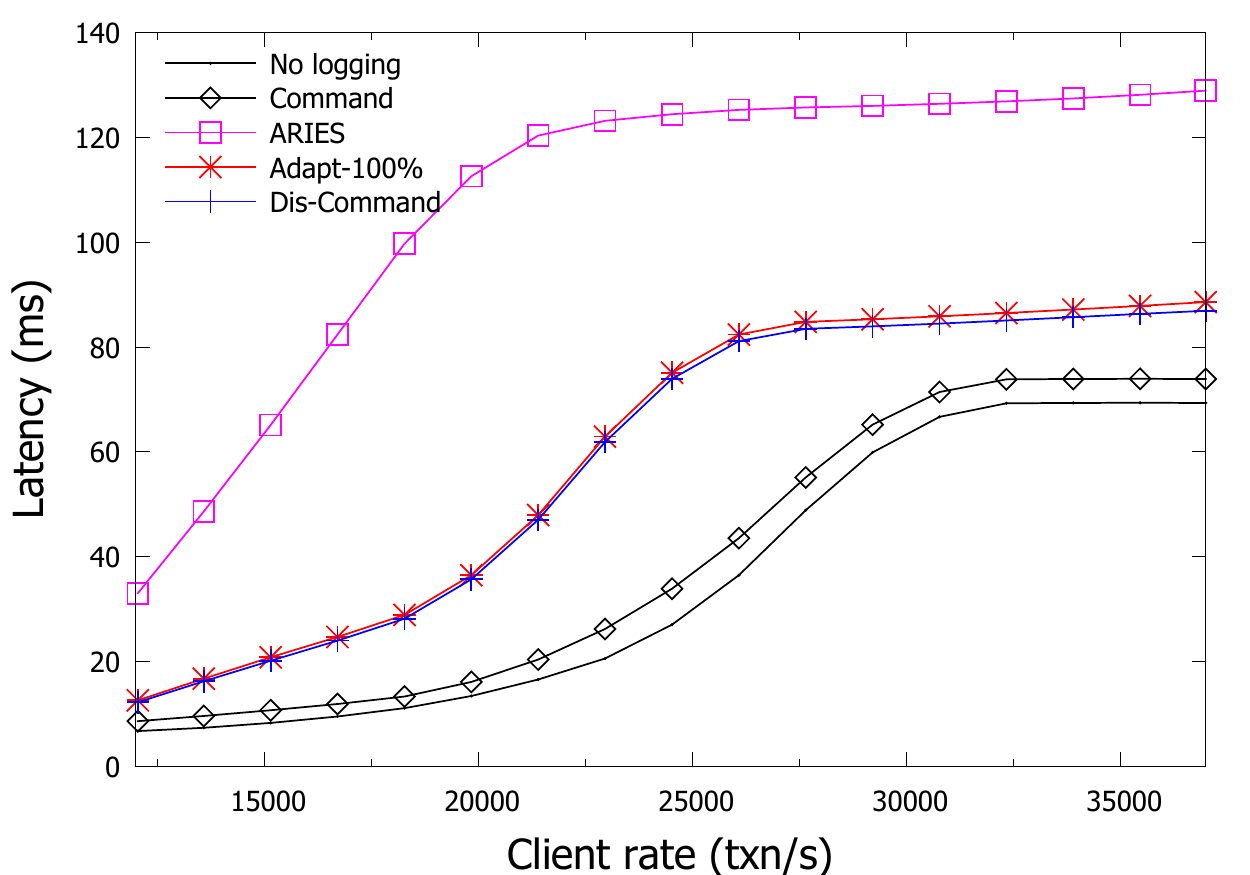}
\caption{\small Latency without distributed transactions\\}
\label{fig:latency_1}
\end{minipage}
\hspace{-3mm}
\hfill
\begin{minipage}[b]{0.23\linewidth}
\includegraphics[width=1\textwidth,height=0.9\textwidth]{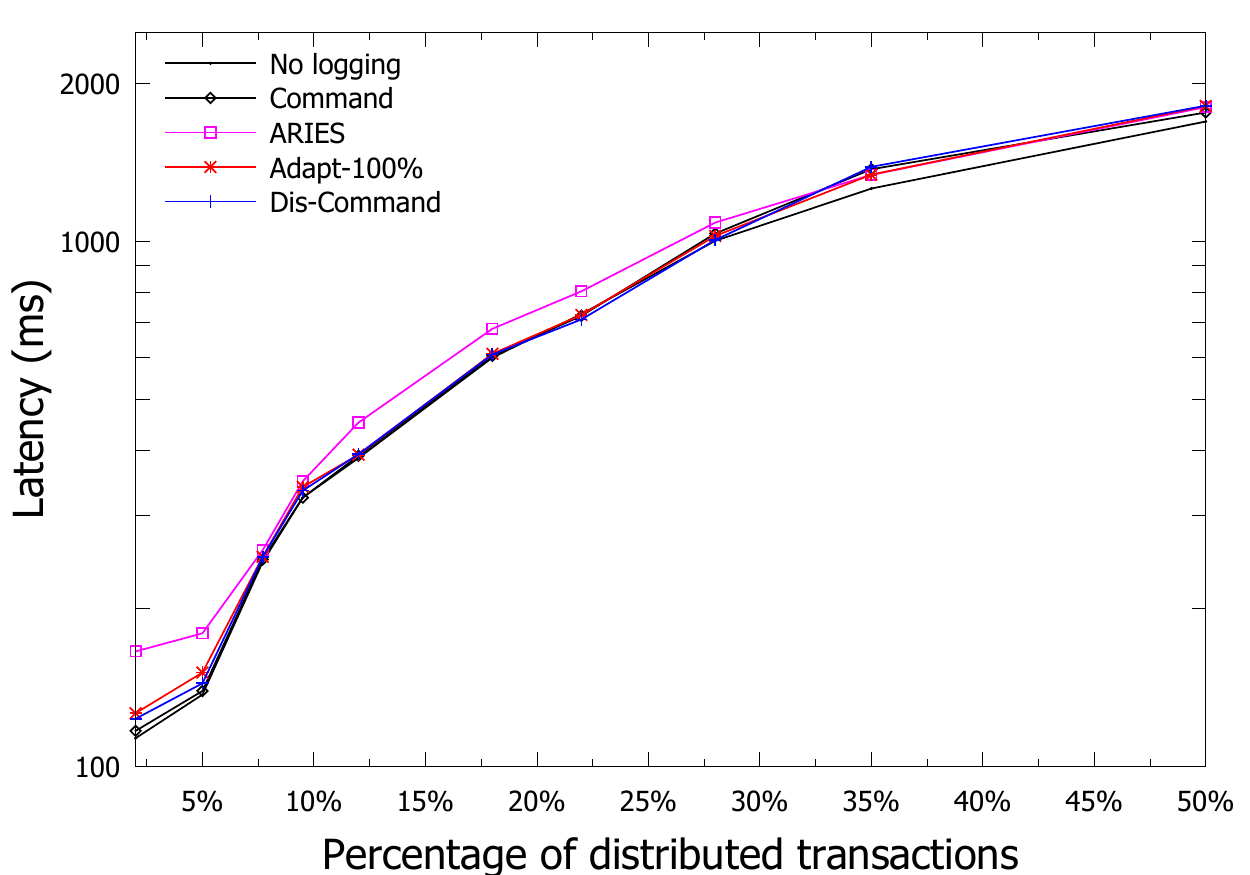}
\caption{\small Latency with distributed transactions (with log scale on y-axis)}\label{fig:dis_lat_1}
\end{minipage}
\vspace {-10pt}
\end{figure*}

\subsubsection{Throughput Evaluation}
Figure \ref{fig:throughout_1} shows the throughput of different approaches when only
local transactions are involved and we vary the client rate, namely the total number of transactions submitted by all client threads per second.
{\color{black}
When the client rate is low,
the system is not saturated and all incoming transactions can be completed
within a bounded waiting time.
}
Although different logging approaches incur different I/O costs,
all logging approaches show a fairly similar performance due to the fact that
I/O is not the bottleneck.
However, as the client rate increases,
{\color{black}
the system with \aries\ logging saturates the earliest at around the input rate
of 20,000 transactions per second.
}
The other approaches (i.e., adaptive logging, distributed logging and command logging),
on the other hand, reach the saturation point around 30000 trasanctions per second which
is slightly lower than the ideal case (represented as no logging approach).
The throughput of distributed command logging is slightly lower than that of command logging
primarily due to the overhead of extra book-keeping involved in distributed command logging.

Figure \ref{fig:dis_thr_1} shows the throughput variation
(with log scale on y-axis)
when there exist distributed transactions.
We set the client rate to 30,000 transactions per second
to keep all sites busy and vary the percentage
of distributed transactions from 0\% to 50\%,
so that the system performance is affected
by both network communications and logging.
To process distributed transactions, 
multiple sites have to cooperate with each other, and
as a result,
the coordination cost typically increases with the number of participating sites.
To guarantee the correctness at the commit phase,
we use the two-phase commit protocol which is supported by the H-store.
In contrast to the local processing shown in Figure \ref{fig:throughout_1},
the main bottleneck of distributed processing gradually shifts
from logging cost to communication cost.
{\color{black}
Compared to local transaction,
distributed transaction always incurs extra network overhead,
with which the effect of logging is less significant.
} 

As shown in Figure \ref{fig:dis_thr_1},
when the percentage of distributed transactions is less than 30\%,
the throughput of the other logging strategies are still 1.4x better
than \aries\ logging.
In this experiment,
the threshold $x$ of adaptive logging is set to 100\%,
where we create  \aries\ logs
for all distributed transactions.
The purpose is to test
the worst performance of adaptive logging.

Command logging is claimed to
be more suitable for local transactions with multiple updates
and \aries\ logging is preferred for distributed
transaction with few updates ~\cite{DBLP:conf/icde/MalviyaWMS14}.
This claim is true in general.
However, since the workload does change over time,
neither command logging nor \aries\ logging can fully
satisfy all access patterns.
On the other hand,
our proposed adaptive logging
has been designed to adapt to the real time variability in workload characteristics.

\subsubsection{Latency Evaluation}
Latency typically exhibits similar trend to that of throughput,
but in the opposite direction,
and
the average latency of different logging strategies
is expected to increase as the client rate increases.
Figure \ref{fig:latency_1} shows that the latency of distributed command logging
is slightly higher than command logging.
However,
it still performs much better than \aries\ logging.
Like other OLTP systems,
H-Store first buffers the incoming transactions in a transaction queue.
The transaction engine will pull them out and process them one by one.
H-Store adopts single-thread mechanism,
in which each thread is responsible for one partition in
order to reduce the overhead of concurrency control.
When the system becomes saturated,
newly arrived transactions need to wait in the queue,
which leads to a higher latency.

Transactions usually commit at the server side which
sends response information back to the client.
Before committing, all log entries are flushed to disk.
The proposed distributed command logging will materialize command log entries and
footprint information before the transactions commit.
When a transaction completes,
it compresses the footprint information 
and as a result,
contributes to a slight delay in response.
However,
the penalty becomes negligible
when many distributed transactions are involved.
{\color{black}
Distributed transaction usually incur a higher latency due to the extra network overhead. In our experiments, group commit is enabled to optimize the disk I/O utility.
}
With an increasing number of distributed transactions,
the latency is less affected by the logging,
all approaches
show a similar performance as shown in Figure \ref{fig:dis_lat_1}. 

\subsubsection{Online Algorithm Cost of Adaptive Logging}
Figure~\ref{fig:online} shows the overhead of the online algorithm. We analyze the computation cost of every minute
by showing the percentages of time taken for
making online decisions and for processing transactions.
The overhead of the online algorithm increases when the system runs for a longer time, because
more transaction information is maintained in the in-memory index.
However, we observe that it takes
only 5 seconds to execute the online algorithm in the 8th minute,
the main bulk of time is still spent on
transaction processing.
Further,
the online decision cost will not grow in an unlimited manner as
it is bounded by the checkpointing interval.
Since the online decision is made before the execution of a transaction,
we could overlap the computation 
while the transaction is waiting in the transaction queue to further reduce
the latency.

\begin{figure*}
\centering
\begin{minipage}[b]{0.32\linewidth}
\includegraphics[width=1\textwidth,height=0.6\textwidth]{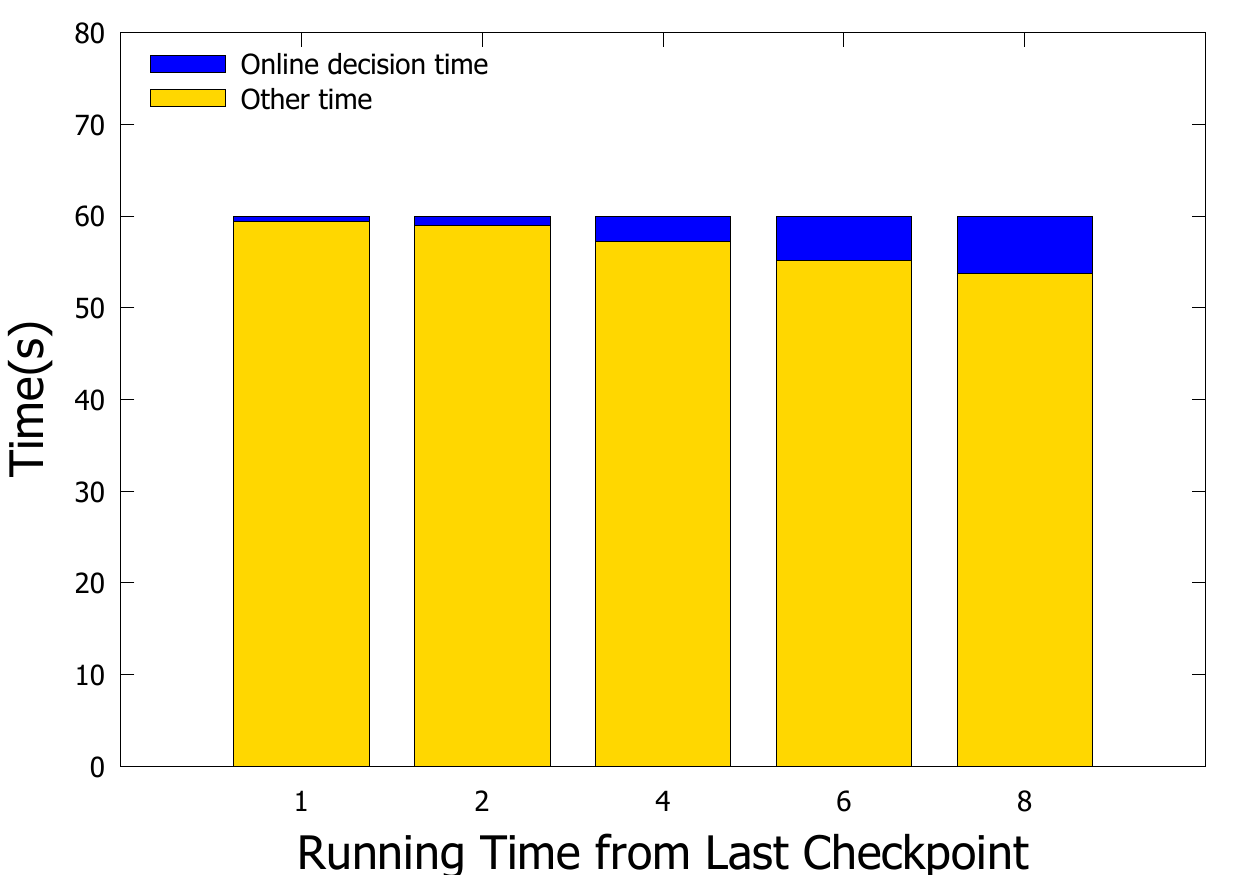}
\caption{\small Cost of online decision algorithm}
\vspace{12pt}
\label{fig:online}
\vspace{-0mm}
\end{minipage}
\hspace{-2mm}
\begin{minipage}[b]{0.64\linewidth}
\begin{subfigure}[b] {.5\linewidth}
  \includegraphics[width=1\textwidth,height=0.6\textwidth]{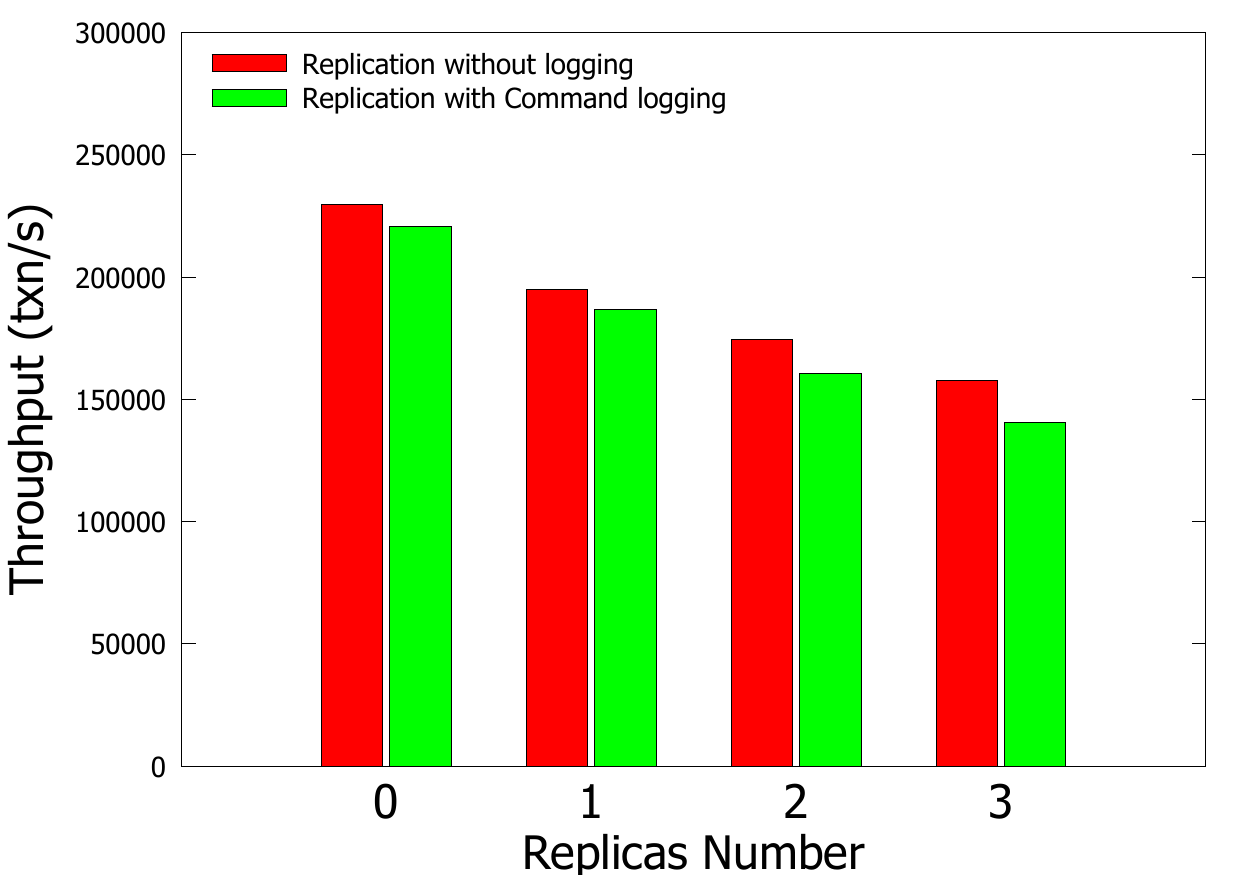}
  \caption{\small Replication with 8 working sites}
  \vspace{10pt}
  \label{fig:rep_1}
\end{subfigure}
\vspace{-5mm}
\begin{subfigure}[b] {.5\linewidth}
 \includegraphics[width=1\textwidth,height=0.6\textwidth]{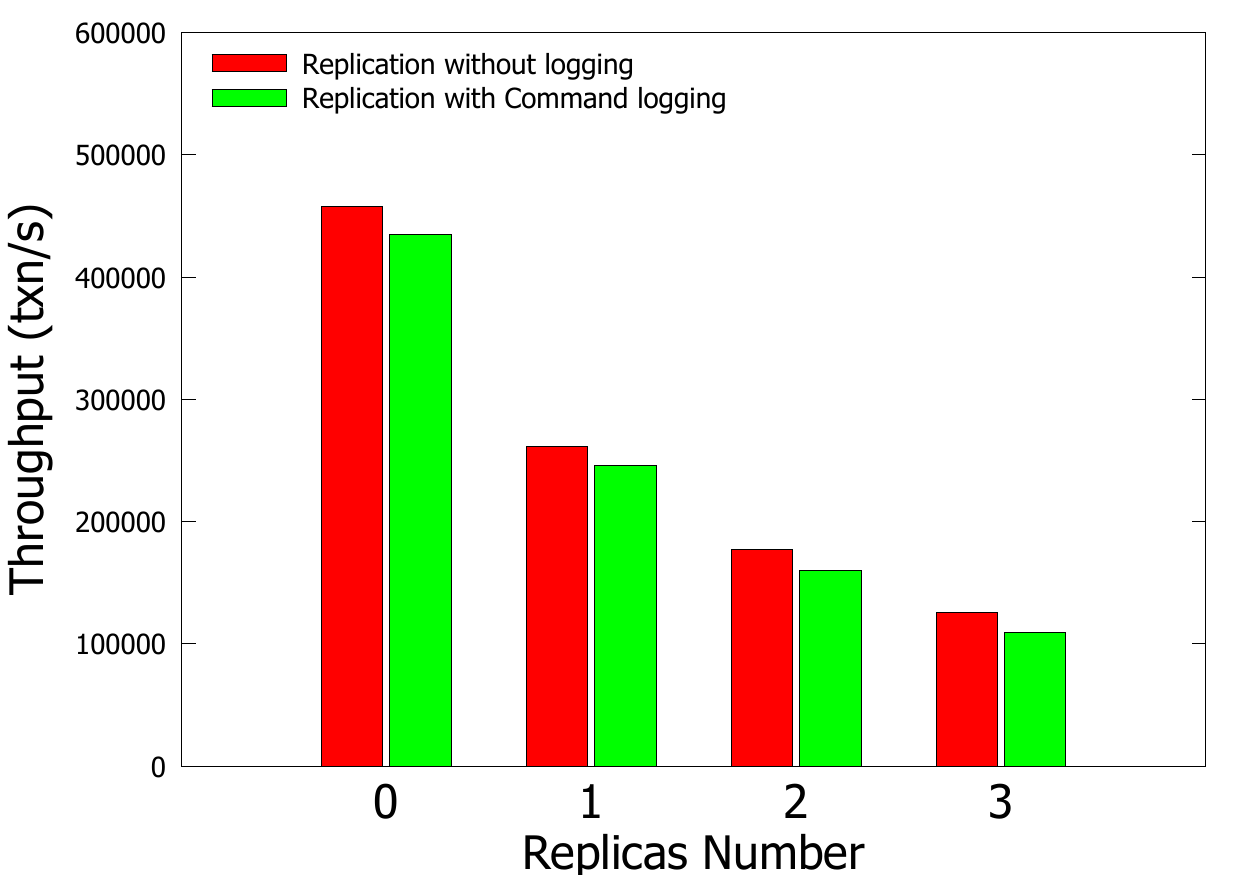}
 \caption{\small Replication with 24 unique sites}
 \vspace{10pt}
 \label{fig:rep_2}
\end{subfigure}
\caption{\small Effects of replication}\label{fig:rep}
\end{minipage}
\vspace{-5mm}
\end{figure*}
\subsubsection{Effect of Replication}
Replication is often used to ensure database correctness,
improve performance, and provide
high availability.
However,
it also incurs high synchronization overhead to achieve strong consistency.
Figure~\ref{fig:rep_1} shows the results with different number of replicas,
where the number of working execution sites is fixed to 8.
{\color{black}
With a fixed number of sites, creating more replicas increases the workload per node,
and more computation resources (e.g., CPU and memory) are used.
The performance drops by about 37\% when there are three replicas.
If available resources are limited,
the system's performance is very sensitive to the number for replicas.
In Figure~\ref{fig:rep_2}, 
we increase the number sites to 24. Namely, $N$ sites are maintaining
the original data, while the other $24-N$ sites are handling the replicas.
In this case, each site will have the same workload as the non-replication case.
However, we find that if 3 replicas are enabled,
the performance degrades by 72.5\% when compared to no replication.
}



\subsection{Recovery Cost Analysis}
In this experiment,
we evaluate the recovery performance of different logging approaches.
We simulate two scenarios.
In the first scenario, we run the system for one minute to process
the transactions and then shut down an arbitrary site to simulate the recovery process.
In the second scenario, each site will process 30,000 transactions
before the process of a random site is terminated forcibly
so that the recovery algorithm can be invoked.
In both scenarios, we measure the elapsed time to recover the failed site.
\begin{figure}[h]
\centering
\vspace {-5pt}
\begin{minipage}[b]{\linewidth}
\includegraphics[width=0.85\textwidth,height=0.45\textwidth]{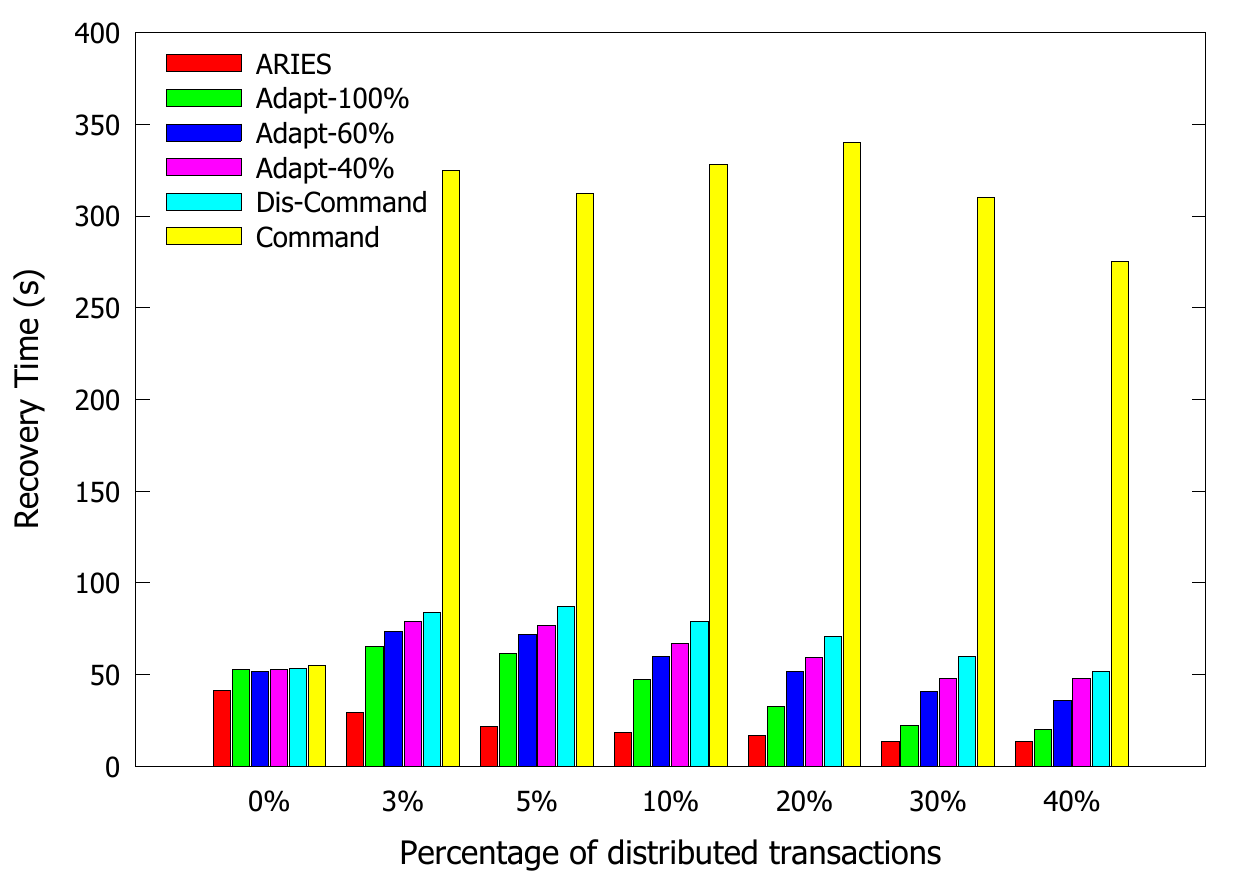}
\vspace{-10pt}
\caption{\small {1 minute after the last checkpoint}}
\label{fig:rec_time1}
\end{minipage}
\vspace{-8pt}
\begin{minipage}[b]{\linewidth}
\includegraphics[width=0.85\textwidth,height=0.45\textwidth]{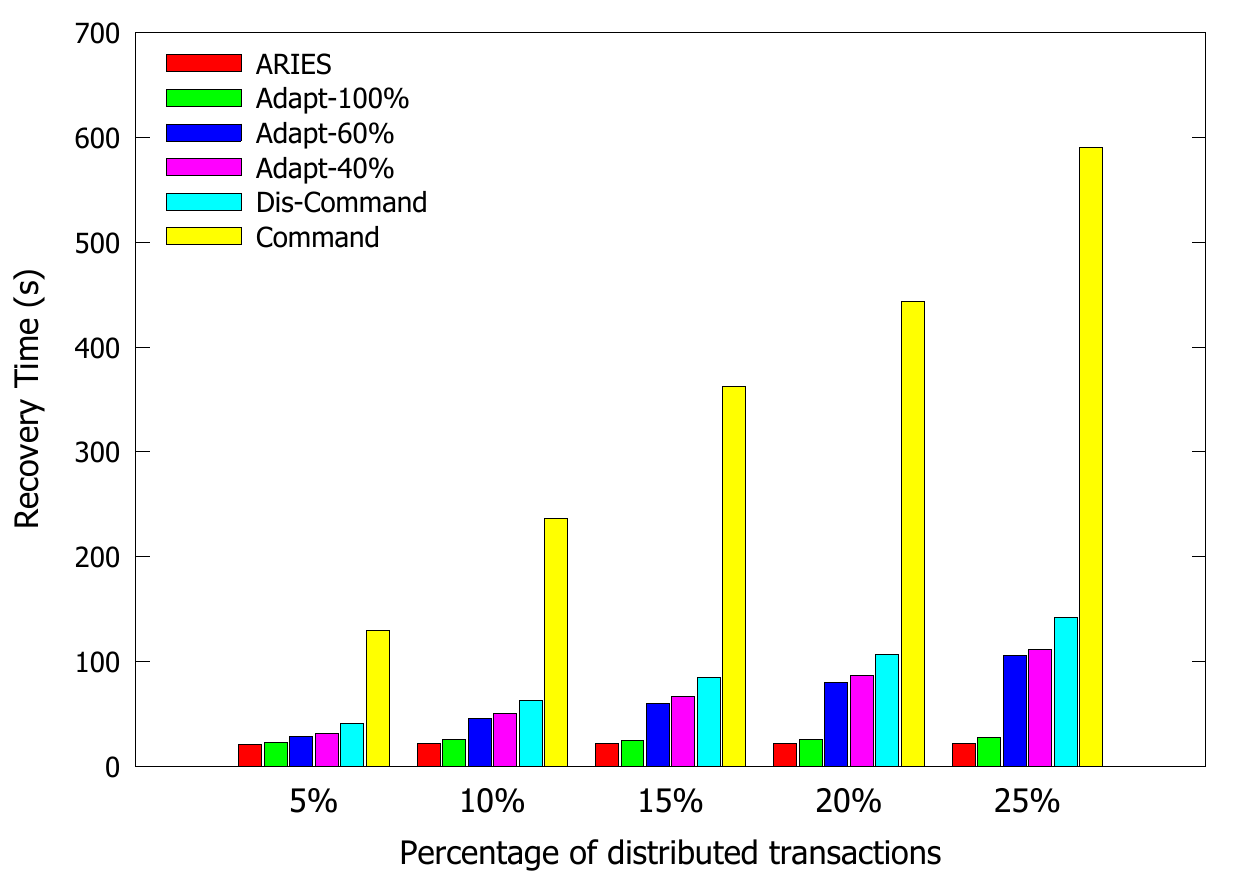}
\vspace{-10pt}
\caption{\small {After 30,000 transactions committed at each site}}\label{fig:rec_time2}
\end{minipage}
\vspace {-10pt}
\end{figure}

\subsubsection{Recovery Evaluation}
{\color{black}
Except for \aries\ logging,
the recovery times of the other methods are affected by two factors,
the number of committed transactions and 
the percentage of distributed transactions.
Figure \ref{fig:rec_time1} and \ref{fig:rec_time2} summarize the recovery times of
the four logging approaches.
Intuitively, the recovery time is proportional to the number of transactions that
must be reprocessed after a failure. 
In Figure \ref{fig:rec_time1},
we note that fewer transactions can be
completed within a given time as the percentage of distributed transactions is increased.
So even though recovering a distributed transaction is costlier, with increased
percentage of distributed transactions there are a fewer number of transactions processed
per unit of time.
Figure~\ref{fig:rec_time1} demonstrates this trade-off in that 
the percentage of distributed transactions does not adversely affect the
recovery times since the cost of recovering distributed transactions is offset
by the reduction in the number of distributed transaction in a fixed unit of time.
For the experiment shown in Figure~\ref{fig:rec_time2},
when we require all sites to complete at least 30,000 transactions,
a higher recovery cost is observed with the increase in the percentage
of distributed transactions.
}

In all cases, \aries\ logging shows the best performance and is not affected
by the percentage of distributed transactions, while
command logging is always the worst. 
Our distributed command logging significantly
reduces the recovery overhead of the command logging,
achieving a 5x improvement.
The adaptive logging further improves the performance by tuning the trade-off
between recovery cost and transaction processing cost as discussed below.

{\bf A{\small{RIES}} logging} supports independent parallel recovery,
since each \aries\ log entry contains one tuple's data
image before and after each operation.
Intuitively, the recovery time of \aries\ logging should be less than
the time interval between checkpointing and the failure time,
{\color{black}
since read operations or transaction logics does not need to be repeated during the recovery.
}
As a fine-grained logging approach,
\aries\ logging is not affected by the percentage of distributed transactions
and the workload skew.
The recovery time is typically proportional to the number of
committed transactions.

{\bf Command logging} incurs much higher overhead when performing a recovery
process involving distributed transactions (even for a small portion, say 5\%).
This observation can be explained
by Figure~\ref{fig:cost_cmdlog} which shows the recovery
time of command logging with one failed site which has 30,000 committed transactions
from the last checkpoint.
The ideal performance of command logging
is achieved by redoing all transactions in all sites without any synchronization.
Of course,
this results in an inconsistent state and
we only use it here to underscore
the overhead of synchronization.
If no distributed transaction is involved,
command logging can provide a similar performance as other schemes,
because dependencies can be resolved within each site.

{\bf Distributed command logging} effectively
{\color{black}
reduces the recovery time compared to the command logging, as shown in Figure~\ref{fig:rec_time1}.
On the other hand, Figure~\ref{fig:rec_time2} shows that the performance of distributed command logging is less sensitive to the percentage of distributed transactions when compared
to command logging.
}
One additional overhead of distributed command logging is the cost of
scanning the footprints to build the dependency graph.  For 1 minute
workload, the time of building dependency graph increases from 2s to
5s when the percentage of distributed transactions ranges from 5\% to
25\%.  Compared to the total recovery cost, the time for building the
dependency graph is fairly negligible.  



\begin{figure}[bt]
\centering
\vspace{0mm}
\begin{minipage}[b]{0.45\linewidth}
\includegraphics[width=1\textwidth,height=0.9\textwidth]{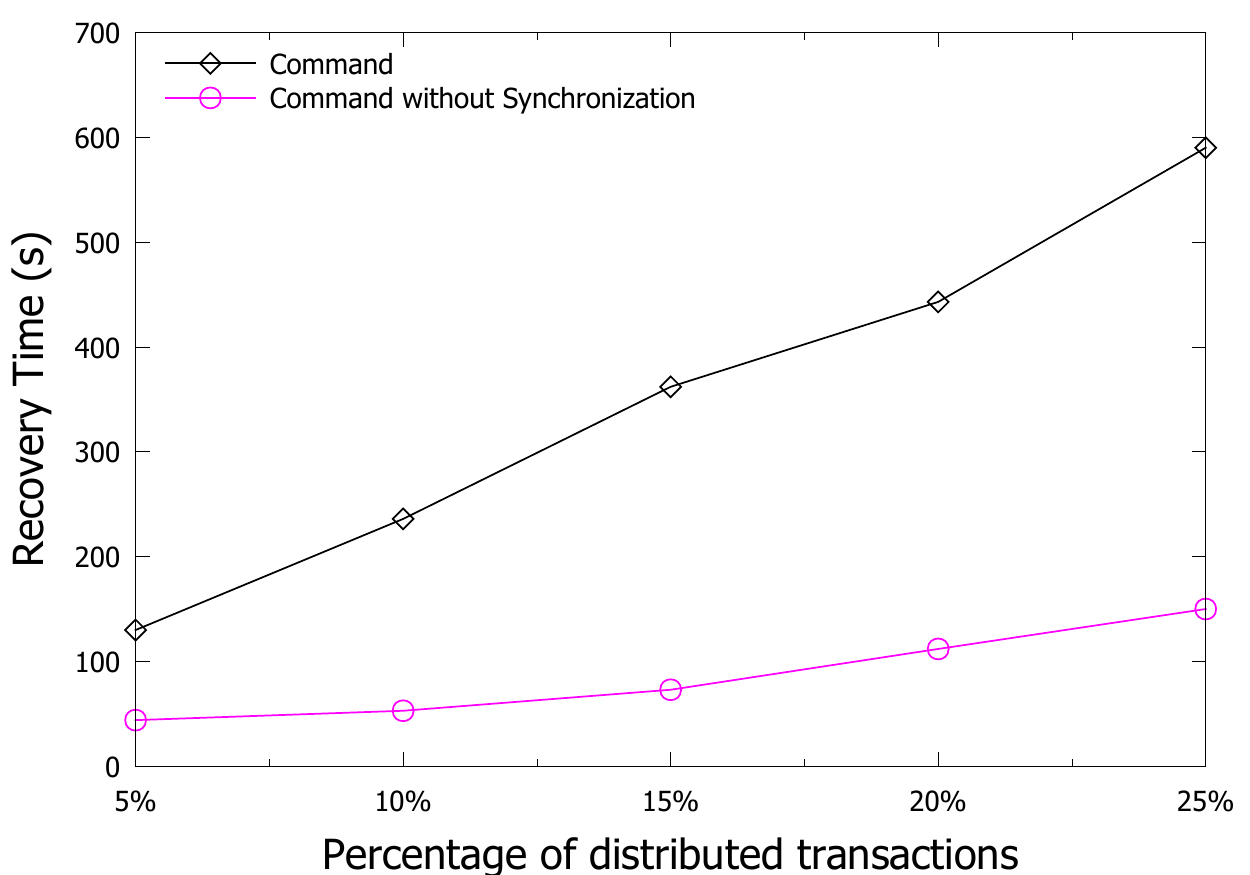}
\caption{\small Synchronization cost  of command logging}
\label{fig:cost_cmdlog}
\end{minipage}
\hspace{-0mm}
\begin{minipage}[b]{0.45\linewidth}
\includegraphics[width=1\textwidth,height=0.9\textwidth]{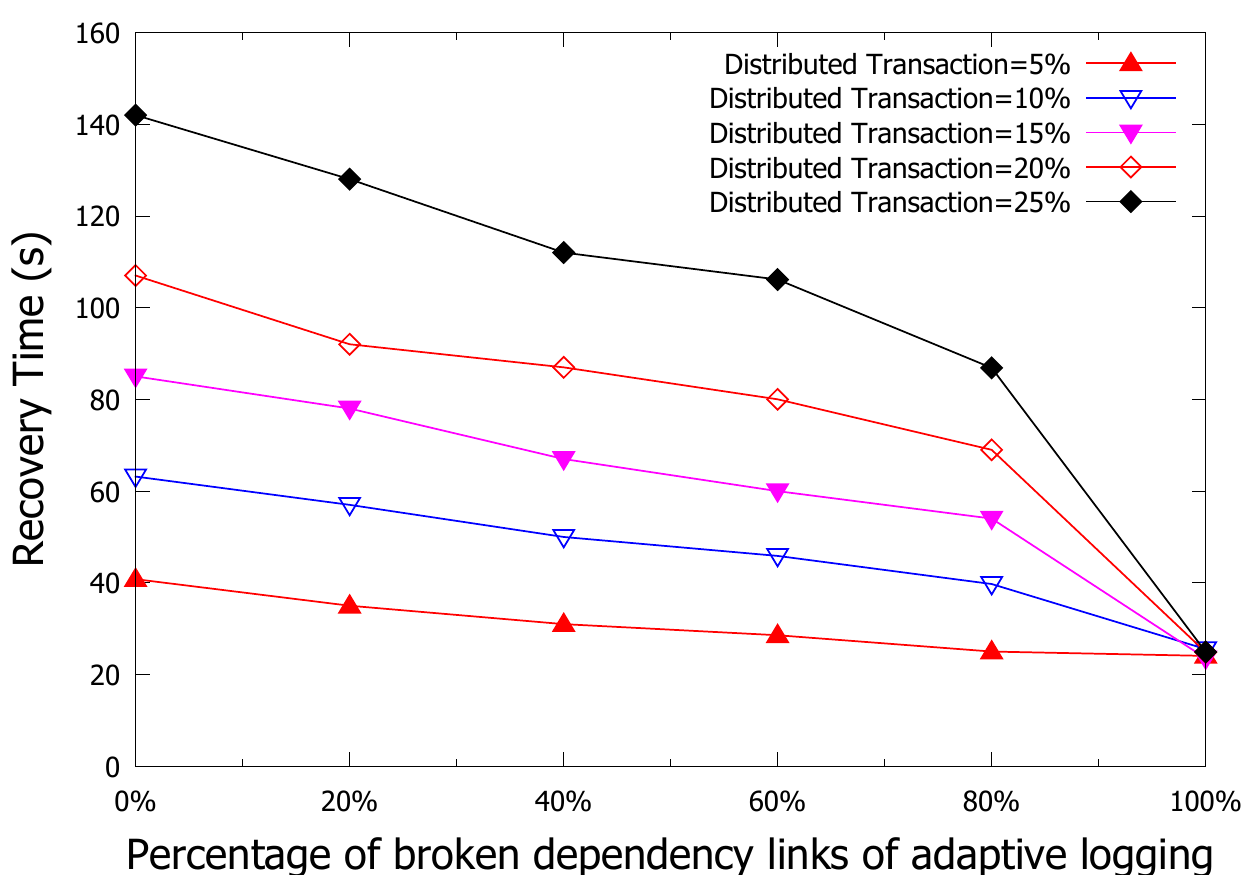}
\caption{\small Recovery performance with varying $x$}
\label{fig:ada_ratio}
\end{minipage}
\vspace{-5mm}
\end{figure}

{\bf Adaptive logging} technique selectively builds the \aries\ log and command log.
To reduce the I/O overhead of adaptive logging, in our online algorithm,
we set a threshold $B_{i/o}$ in online algorithm.
So at most,
$N=\frac{B_{i/o}}{W^{aries}}$ \aries\ logs can be created.
In this experiment,
we use a dynamic threshold, by setting
$N$ as $x$ percentage of the total number of distributed transactions.
In Figure \ref{fig:rec_time1} and \ref{fig:rec_time2},
$x$ is set as 40\%, 60 \% or 100\% to tune the recovery cost
and transaction processing cost. 
In all the tests,
the recovery
performance of adaptive logging is much better
than command logging.
It only performs slightly worse than the pure \aries\ logging.
As $x$ increases, more
\aries\ logs are created by adaptive logging, which results in the reduction of
the recovery time.
In the extreme case,
we create \aries\ log for every distributed transaction by setting $x=100\%$.
Then, all dependencies of distributed transactions are
resolved using \aries\ logs, and
each
site can process its recovery independent of the others.

Figure \ref{fig:ada_ratio} shows the effect of
$x$ on the recovery performance.
We vary the percentage of distributed
transactions and show the results with different $x$ values.
When $x=100\%$,
the recovery times are the same for all,
independent of the percentage of distributed transactions,
because all dependencies have been resolved.
On the contrary,
adaptive logging will degrade to distributed command logging,
if we set $x=0$.
In this case,
more distributed transactions result in higher recovery cost.

{\color{black}
Table~\ref{tb:number} shows the number of transactions that are reprocessed during the recovery in Figure~\ref{fig:rec_time2}.
Compared to command logging, distributed command logging and adaptive logging
efficiently reduce the number of transactions that need to be reprocessed.
}
\begin{table}[h]
\scriptsize
\caption{Number of reprocessed transactions}\label{tb:number}
\vspace{-5pt}
\hspace{-15pt}
\vspace{-5pt}
\begin{tabular}{|c|c|c|c|c|c|}
\hline
{\bf Percentage} & {\bf Command} & {\bf Dis-Command} & {\bf Adapt-40\%} & {\bf Adapt-60\%} & {\bf Adapt-100\%} \\ \hline 
0\%   & 30031 & 30015 & 30201 & 30087 & 30076 \\ \hline 
5\%   & 239321 & 35642 & 33742 & 32483 & 29290 \\ \hline 
10\%  & 240597 & 39285 & 36054 & 34880 & 30674 \\ \hline 
15\%  & 240392 & 42979 & 39687 & 37496 & 32201 \\ \hline 
20\%  & 239853 & 48132 & 43808 & 40912 & 33994 \\ \hline 
25\%  & 240197 & 57026 & 50465 & 46095 & 35617 \\ \hline 
\end{tabular}
\vspace{-10pt}
\end{table}

\subsubsection{Overall Performance Evaluation}
The intuition of the adaptive logging approach is
to balance the tradeoff between recovery and
transaction processing time.
It is widely expected that
when commodity servers are used in a large number,
failures are no longer an exception~\cite{vishwanath2010characterizing}.
That is, the system must be able to recover efficiently when a failure occurs
and provide a good overall performance.
In this set of experiments,
we measure the overall performances of different approaches.
In particular, we run the system for three hours and intentionally
shut down a random node based on a predefined failure rate.
The system will iteratively process transactions
and perform recovery,
and a new checkpoint is created every 10 minutes.
Then, the total throughput of the entire system is computed as the average number of transactions
processed per second in the three hours.

We show the total throughput for varying failure rate
from Figure \ref{fig:overall_5_response} to Figure \ref{fig:overall_20_response} with three different mixes of distributed transactions.
\aries\ logging is superior to
the other approaches when the failure rate is very high (e.g., there is one failure every 5 minutes).
When the failure rate is low,
distributed command logging shows the best performance,
because it is just
slightly slower than command logging for transaction processing,
but recovers much faster than command logging.
As the failure rate drops,
Adapt-100\% approach cannot provide a comparable performance
to command logging,
because Adapt-100\% creates the \aries\ log for every distributed transaction
which is too costly in transaction processing.

\begin{figure*}[bt]
\centering
\begin{minipage}[b]{0.75\linewidth}
\begin{subfigure}[b]{0.32\linewidth}
\includegraphics[width=1\textwidth,height=0.9\textwidth]{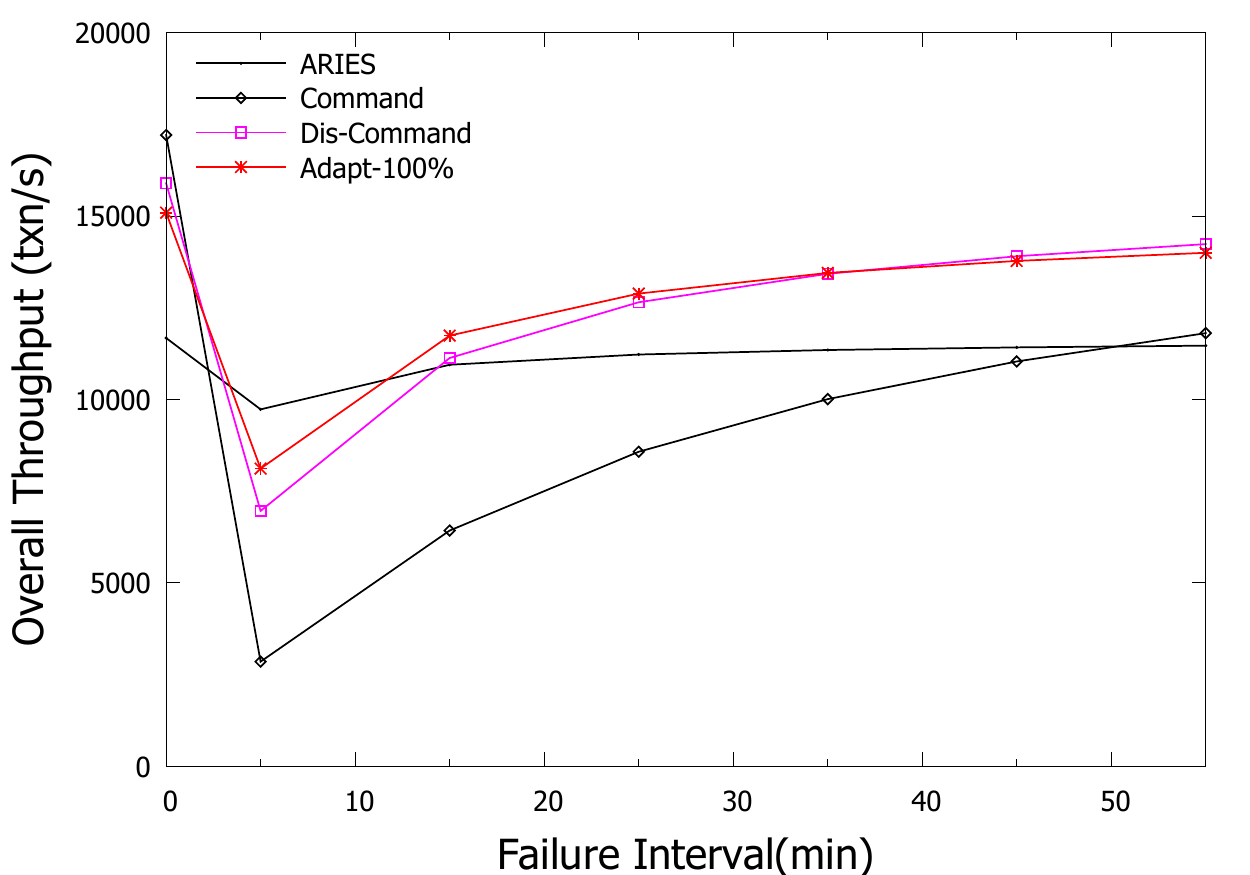}
\caption{Overall throughput with 5\% distributed transactions}
\label{fig:overall_5_response}
\end{subfigure}
\hspace{-3mm}
\hfill
\begin{subfigure}[b]{0.32\linewidth}
\includegraphics[width=1\textwidth,height=0.9\textwidth]{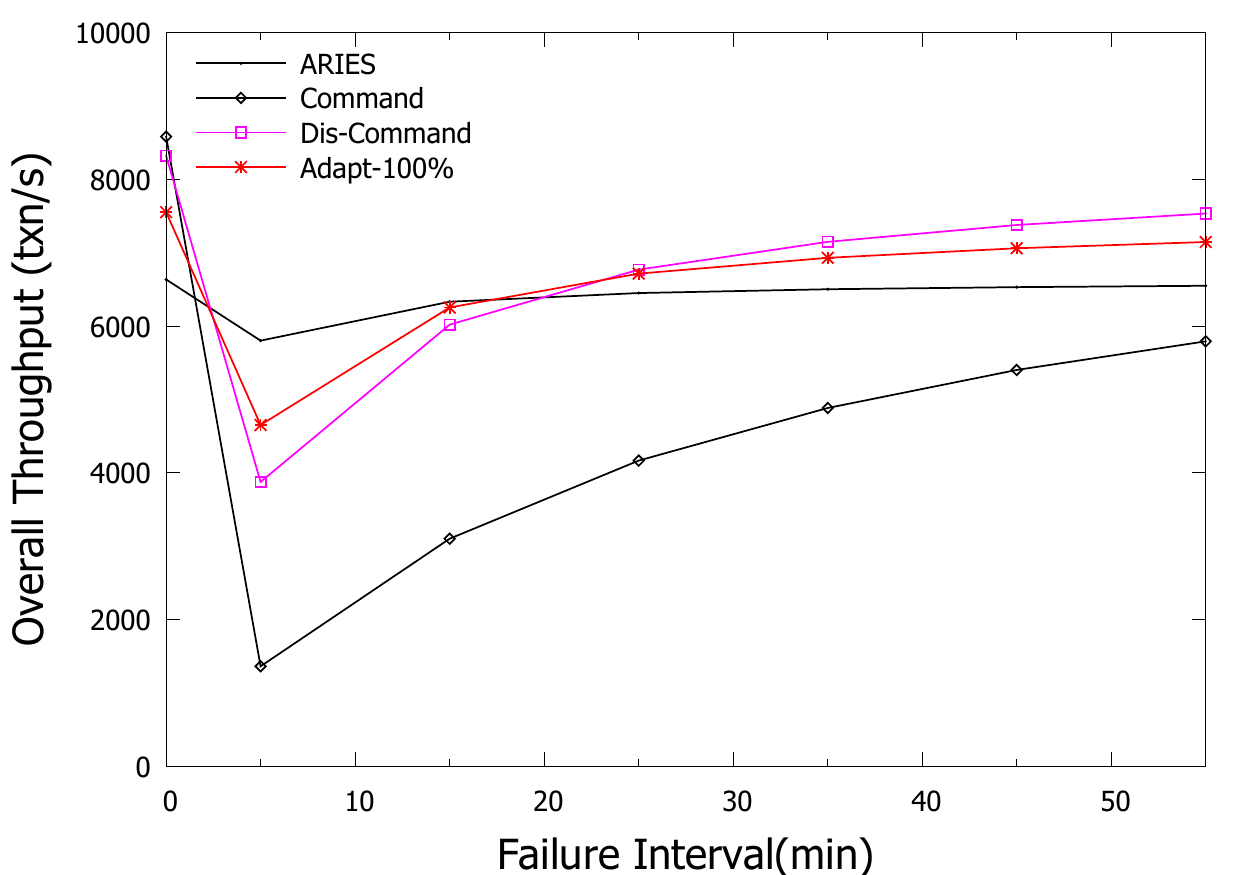}
\caption{Overall throughput with 10\% distributed transactions}
\label{fig:overall_10_response}
\end{subfigure}
\hspace{-2mm}
\hfill
\begin{subfigure}[b]{0.32\linewidth}
\includegraphics[width=1\textwidth,height=0.9\textwidth]{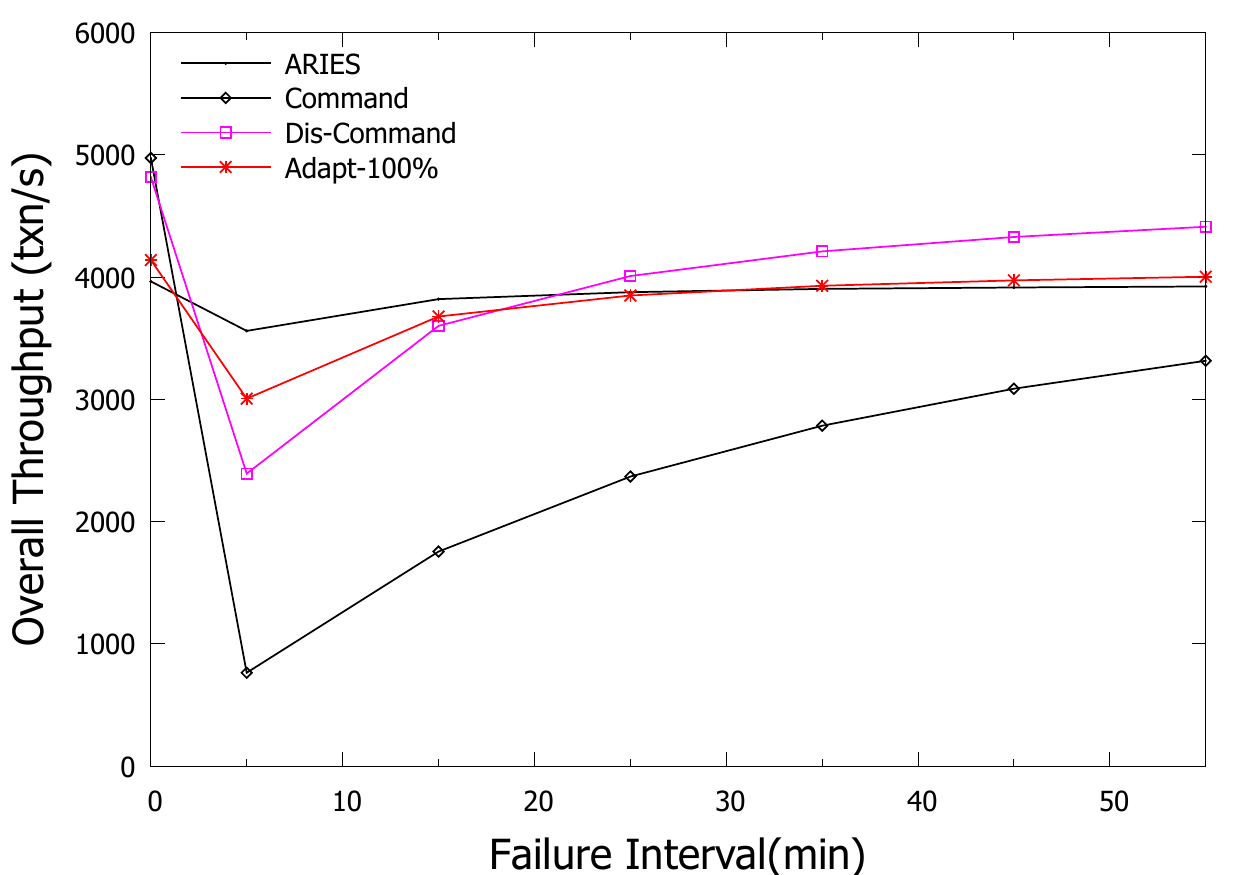}
\caption{Overall throughput with 20\% distributed transactions}
\label{fig:overall_20_response}
\end{subfigure}
\vspace{-7pt}
\caption{Overall performance evaluation}
\label{fig:overall}
\end{minipage}
\hspace{-3mm}
\hfill
\begin{minipage}[b]{0.24\linewidth}
\includegraphics[width=1\textwidth,height=0.9\textwidth]{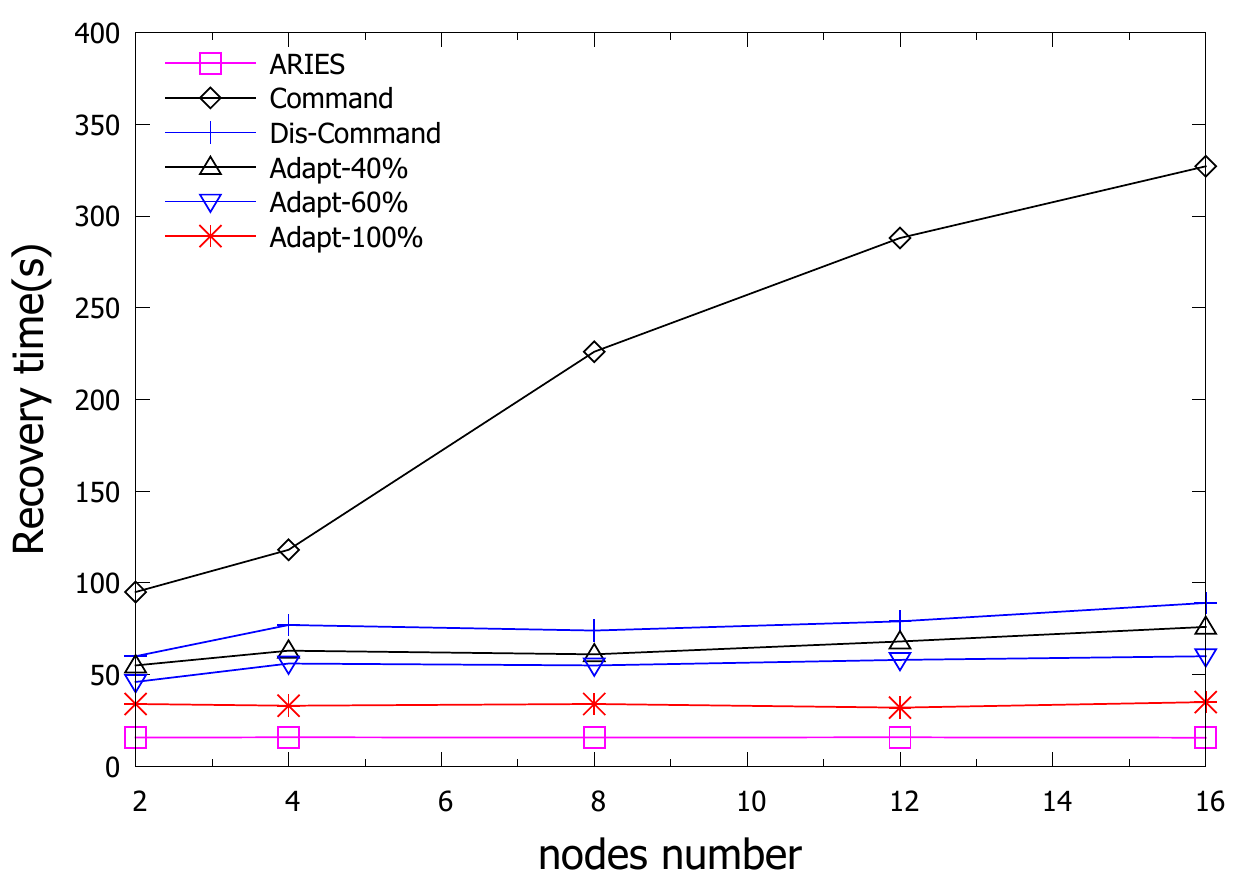}
\caption{Recovery time V.S. node number with distributed transactions}\label{fig:scale_1}
\vspace{-3pt}
\end{minipage}
\vspace {-10pt}
\end{figure*}

\subsubsection{Scalability}
{\color{black}
In this experiment, we evaluate the scalability of our proposed approaches.
In Figure \ref{fig:scale_1},
each site processes at least 30,000 transactions before we randomly terminate one site (other
sites will detect it as a failed site).
The percentage of distributed transactions is 10\% which are uniformly distributed among all sites.
We observe that command logging is not scalable,
as the recovery time is linear to the number of sites,
because all sites need to reprocess their lost transactions.
The recovery cost of distributed command logging increases
by about 50\% when we increase the number of sites from 2 to 16.
The other logging approaches show a scalable recovery performance.
Adaptive logging selectively creates \aries\ logs to break dependency relations among compute nodes. The number of transactions which are required to be reprocessed
is greatly reduced during recovery.}

\vspace{-8pt}
\section{Related Work}
\aries\cite{mohan1992aries} logging is widely adopted for recovery
in traditional disk-based database systems.
As a fine-grained logging strategy,
\aries\ logging needs to construct log records for each modified tuple.
Similar techniques are applied to in-memory database systems\cite{eich1986main,jagadish1994dali,186198,jagadish1993recovering}.

In \cite{DBLP:conf/icde/MalviyaWMS14},
the authors argue that for in-memory systems,
since the whole database is maintained in memory,
the overhead of \aries\ logging cannot be ignored.
They proposed a different kind of coarse-grained logging strategy called command logging.
It only records transaction's name and parameters instead of concrete tuple modification information.

\aries\ log records contain the tuples' old values and new values. Dewitt et al\cite{dewitt1984implementation}
try to reduce the log size by only writing the new values to log files.
However, log records without old values cannot support undo operation.
So it needs large enough stable memory which can hold the complete log records for active transactions.
They also try to write log records in batch to optimize the disk I/O performance.
Similar techniques such as group commit\cite{hagmann1987reimplementing} are also explored in modern database systems.

Systems\cite{PostgreSQL} with asynchronous commit strategy allow transactions
to complete without waiting log writing requests to finish.
This strategy can reduce the overhead of log flush to an extent.
But it sacrifice database's durability, since the states of
the committed transactions can be lost when failures happen.

Lomet et al\cite{lomet2011implementing} propose a logical logging strategy.
The recovery phase of \aries\ logging combines physiological redo and logical undo.
This work extends \aries\ to work in a logical setting.
This idea is used to make \aries\ logging more suitable for in-memory database system.
Systems like \cite{DBLP:journals/pvldb/KallmanKNPRZJMSZHA08,kemper2011hyper} adopt
this logical strategy.

If non-volatile RAM is available, database systems\cite{li1993post} can use it to do
some optimizations at the runtime to reduce the log size by using shadow pages for updates.
With non-volatile RAM, recovery algorithms proposed by Lehman and Garey\cite{lehman1987recovery}
can then be applied.

There are many research efforts\cite{186198,cao2011fast,dewitt1984implementation,pu1986fly,rosenkrantz1978dynamic,salem1989checkpointing} devoted to efficient checkpointing for in-memory database systems.
Recent works such as\cite{186198,cao2011fast} focus
on fast checkpointing to support efficient recovery.
Usually checkpointing techniques need to combine with logging techniques
and complement with each other to realize reliable recovery process.
Salem et al\cite{salem1990system} survey many checkpointing techniques,
which cover both inconsistent and consistent checkpointing with different
logging strategies.

Johnson et al\cite{johnson2010aether} identify logging-related impediments to
database system scalability. The overhead of log related locking/latching contention
decreases the performance of the database systems, since transactions need to
hold locks while waiting for the log to write.
Works such as\cite{johnson2010aether,pandis2010data,pandis2011plp}
try to make logging more efficient by reducing the effects of locking contention.

RAM-Cloud\cite{ongaro2011fast},
a key-value storage for large-scale applications,
replicates node's memory across nearby disks.
It is able to support very fast recovery
by careful reconstructing the failed data
from many other healthy machines.

\section{Conclusion}
{\color{black}
In the context of in-memory databases, Compared to
command logging\cite{DBLP:conf/icde/MalviyaWMS14} shows a much better performance for transaction processing
compared to the traditional write-ahead logging (\aries\ logging\cite{mohan1992aries}).
However, the trade-off is that command logging can significantly increase recovery times in the case of a failure.
The reason is that command logging redoes all transactions in the log
since the last checkpoint in a serial order. 
To address this problem,
we first extend command logging to distributed systems to enable
all the nodes to perform their recovery in parallel. 
We identify the transactions
involved in the failed node by analyzing the dependency relations and only redo those involved transactions to reduce
the recovery overhead. 
We find that the recovery bottleneck of command logging
is the synchronization process to resolve data dependency. 
Consequentially, we design a novel adaptive
logging approach to achieve an optimized trade-off
between the performance of transaction processing and recovery.
Our experiments on H-Store show
that adaptive logging can achieve a 10x boost for
recovery and its transaction throughput is comparable to command logging.
}


\bibliographystyle{abbrvnat}
\bibliography{adaptivelogging}

\end{document}